\documentclass[aos, preprint]{imsart}

\RequirePackage[OT1]{fontenc}
\RequirePackage{amsthm,amsmath}
\usepackage{array}
\usepackage{mathrsfs}
\usepackage{amssymb,bbm}
\usepackage{amsfonts}
\usepackage{amsmath,amssymb}
\usepackage{graphicx}
\usepackage{amscd}
\usepackage{color}
\usepackage{mathrsfs}
\usepackage{latexsym,bm}
\usepackage{array}
\usepackage{threeparttable}

\usepackage{threeparttable}
\usepackage{rotating}
\usepackage{booktabs}

\usepackage{rotating}

\usepackage[plain,noend]{algorithm2e}

\usepackage{makecell}
\usepackage{subfigure}
\usepackage{diagbox}
\usepackage{enumitem}
\usepackage{bm}

\DeclareFontFamily{U}{mathx}{\hyphenchar\font45}
\DeclareFontShape{U}{mathx}{m}{n}{
      <5> <6> <7> <8> <9> <10>
      <10.95> <12> <14.4> <17.28> <20.74> <24.88>
      mathx10
      }{}
\DeclareSymbolFont{mathx}{U}{mathx}{m}{n}
\DeclareMathAccent{\widecheck}{\mathalpha}{mathx}{"71}

\DeclareMathAccent{\widecheck}{\mathalpha}{mathx}{"71}

\startlocaldefs \numberwithin{equation}{section} \theoremstyle{it}
\newtheorem{thm}{Theorem}[section]
\newtheorem{lem}{Lemma}[section]
\newtheorem{ass}{Assumption}[section]
\newtheorem{cor}{Corollary}[section]

\newtheorem{rem}{Remark}
\endlocaldefs

\newcommand{\argmin}{\operatornamewithlimits{argmin}}

\newcommand{\f}{\frac}

\begin{document}

\begin{frontmatter}
\title{Hybrid Quantile estimation for asymmetric power GARCH models}

\thankstext{}{JEL classification: C01, C22, C58}

\begin{aug}
\author{\fnms{Guochang} \snm{Wang}\thanksref{m1}\ead[label=e1]{twanggc@jnu.edu.cn}},
\author{\fnms{Ke} \snm{Zhu}\thanksref{m2}\ead[label=e2]{mazhuke@hku.hk}},
\author{\fnms{Guodong} \snm{Li}\thanksref{m2}\ead[label=e3]{gdli@hku.hk}}
\and
\author{\fnms{Wai Keung} \snm{Li}\thanksref{m2,m3}
\ead[label=e4]{waikeungli@eduhk.hk}}


\affiliation{Jinan University\thanksmark{m1}, The University of Hong Kong\thanksmark{m2}\\
\and The Education University of Hong Kong\thanksmark{m3}}

\address{
College of Economics\\
Jinan University\\
Guangzhou, China\\
\printead{e1}\\
}

\address{
Department of Statistics \& Actuarial Science\\
The University of Hong Kong\\
Hong Kong\\
\printead{e2}\\
\phantom{E-mail:\ }\printead*{e3}\\
}

\address{
 Department of Mathematics and Information Technology\\
The Education University of Hong Kong\\
Hong Kong\\
\printead{e4}\\
}
\end{aug}

\begin{abstract}
Asymmetric power GARCH models have been widely used to study the higher order moments of financial returns, while their quantile estimation has been rarely investigated. This paper introduces a simple monotonic transformation on its conditional quantile function to make the quantile regression tractable. The asymptotic normality of the resulting quantile estimators is established under either stationarity or non-stationarity. Moreover,  based on the estimation procedure, new tests for strict stationarity and asymmetry
are also constructed. This is the first try of the quantile estimation for non-stationary ARCH-type models in the literature. The usefulness of the proposed methodology is illustrated by simulation results and real data analysis.\\
\end{abstract}



\end{frontmatter}

\section{Introduction}
Since the seminal work in Engle (1982) and Bollerslev (1986),
the generalized autoregressive conditional heteroskedasticity (GARCH) model has been widely used to capture the volatility clustering of financial data; see, e.g., Francq and Zako\"{i}an (2010) for an overview.  Financial data  are well known to exhibit conditional asymmetric features, in the sense that large negative returns tend to have more impact on future volatilities than large positive returns of the same magnitude. This stylized fact, which is known as the leverage effect, was first documented by Black (1976), and leads to many variants of the classical GARCH model (see, e.g., Higgins and Bera, 1992; Li and Li, 1996; Zhu et al., 2017). Among the existing asymmetric ARCH-type models,
 the first order asymmetric power-transformed GARCH (PGARCH) model proposed by Pan et al. (2008) is often used in applications, and it is defined by
 \begin{flalign}
\epsilon_t=h_t^{1/\delta}\eta_t,\,\,\,h_t=\omega_0+\alpha_{0+}(\epsilon_{t-1}^+)^{\delta}+\alpha_{0-}(-\epsilon_{t-1}^-)^{\delta}+\beta_0 h_{t-1},
\label{1.1}
\end{flalign}
where $\delta$ is a given positive constant exponent, $\omega_0>0$, $\alpha_{0+}\geq 0, \alpha_{0-}\geq 0$, $\beta_0\geq 0$, and $\{\eta_t\}$ is a sequence of independent and identically distributed (i.i.d.) random variables. Here, the notations
$x^+=\max(x,0)$ and $x^-=\min(x,0)$ are used. Model (1.1) is
motivated by the Box-Cox transformation, and it covers the classical GARCH model in Engle (1982) and Bollerslev (1986),
the absolute value GARCH in Taylor (1986),
the GJR model in Glosten et al. (1993), the threshold GARCH model in Rabemananjara and Zako\"{i}an (1993),
the PARCH model in
Hwang and Kim (2004), and many others.

Following H\"{o}rmann (2008), model (\ref{1.1}) is stationary if and only if the top Lyapunov exponent
$\gamma_0<0$, where
\begin{flalign}\label{1.2}
\gamma_0=E\log a_0(\eta_t),\,\,\,a_0(x)=\alpha_{0+}(x^+)^{\delta}+\alpha_{0-}(-x^-)^{\delta}+\beta_0.
\end{flalign}
By assuming $\eta_t$ follows a standard normal distribution, the Gaussian quasi-maximum likelihood estimator (QMLE) of
model (\ref{1.1}) was studied in Pan et al. (2008) and Hamadeh and Zako\"{i}an (2011) for $\gamma_0<0$, and
 Francq and Zako\"{i}an (2013a) for $\gamma_0\geq0$. Although the Gaussian QMLE has some desired asymptotic properties,
 it overlooks a crucial practical feature that the
quantile structure of the financial data actually varies in shape across the quantile levels (Engle and Manganelli, 2004).
 Nowadays, the estimation of the conditional quantile becomes increasingly important for the financial data, since it is related to
 the quantile-based risk measures such as
 Value-at-Risk (VaR) and Expected Shortfall (ES), which are implemented worldwide in financial market
regulation and banking supervision. However, only few attempts have been made to study the quantile estimation for model (\ref{1.1}), especially when $\gamma_0\geq0$.

This paper contributes to the literature in two aspects. First, we extend the idea of Zheng et al. (2018) to construct a hybrid conditional quantile estimator of $\epsilon_t$ in model (1.1). To elaborate this idea, we let $\theta_0=(w_0, \alpha_{0+}, \alpha_{0-}, \beta_0)'$ and
$\theta_{\tau0}=b_{\tau}\theta_0$, where $\tau\in(0, 1)$ is the given quantile level,
 $b_{\tau}=T(Q_{\tau,\eta})$,
$Q_{\tau,\eta}$ is the $\tau$th quantile of $\eta_t$, and $T(x)=|x|^{\delta}\mbox{sgn}(x)$ is a
given monotonic transformation function. Then, the $\tau$th quantile of the transformed data $y_t=T(\epsilon_t)$ conditional on $\mathcal{F}_{t-1}$ is
\begin{flalign}\label{1.3}
Q_{\tau}(y_t|\mathcal{F}_{t-1})=b_{\tau}(\omega_0+\alpha_{0+}(\epsilon_{t-1}^+)^{\delta}+\alpha_{0-}(-\epsilon_{t-1}^-)^{\delta}+\beta_{0}h_{t-1})=\theta_{\tau0}'z_t,
\end{flalign}
and the $\tau$th quantile of the original data $\epsilon_t$ conditional on $\mathcal{F}_{t-1}$ is
\begin{flalign}\label{1.4}
Q_{\tau}(\epsilon_t|\mathcal{F}_{t-1})=T^{-1}\big(Q_{\tau}(y_t|\mathcal{F}_{t-1})\big),
\end{flalign}
where $z_t=(1,(\epsilon_{t-1}^+)^{\delta},(-\epsilon_{t-1}^-)^{\delta},h_{t-1})'$, $\mathcal{F}_{t}$ is
the $\sigma$-field generated by $\{\epsilon_t,\epsilon_{t-1},...\}$, and
$T^{-1}(x)=|x|^{1/\delta}\mbox{sgn}(x)$.
The result (\ref{1.3}) implies that
$Q_{\tau}(y_t|\mathcal{F}_{t-1})$ is linear in terms of $z_t$, and hence if $z_t$ is observable,
$\theta_{\tau0}$ can be easily estimated by the regression quantile estimation.
With this quantile estimator of $\theta_{\tau0}$, then $Q_{\tau}(y_t|\mathcal{F}_{t-1})$ can be estimated
via (\ref{1.3}), leading to an estimator of $Q_{\tau}(\epsilon_t|\mathcal{F}_{t-1})$ according to (\ref{1.4}).
However, $z_t$ contains an unobservable $h_{t-1}$, which
 has a recursive form, adding difficulty to the theoretical derivation and numerical optimization.
 To circumvent this difficulty, we replace $h_{t-1}$ by some initial estimators to calculate the quantile estimator of
 $\theta_{\tau0}$; see also Xiao and Koenker (2009), So and Chung (2015) and  Zheng et al. (2018).
  Indeed, Zheng et al. (2018) estimated $h_{t-1}$ based on the Gaussian QMLE, which needs
 $E\eta_{t}^{4}<\infty$ in theory. To relieve the moment condition of $\eta_t$, we estimate $h_{t-1}$ by using the
 generalized QMLE (GQMLE) in  Francq and Zako\"{i}an (2013b), and our theory only requires $E|\eta_{t}|^{2r}<\infty$,
 where $r$ is a user-chosen positive number, indicating the estimation method used.
 Note that there is a vast literature on the estimation of conditional quantile for financial data, and two leading examples are the filtered historical simulation (FHS) method (Barone-Adesi et al., 1998; Barone-Adesi and Giannopoulos, 2001; Kuester et al.,
2006) and the conditional auto-regressive VaR-method called ``CAViaR'' (Engle and Manganelli,
2004).
As argued in Zheng et al. (2018), the hybrid conditional quantile estimation method
combines the advantages of both FHS and CAViaR approaches, since it can
exploit the ARCH-type structure in both the global
estimation of the volatility and the local estimation of quantiles.

Second, we study the asymptotic properties of the quantile estimator of $\theta_{\tau0}$. Denote
$\theta_{\tau0}=(\omega_{\tau0}, \vartheta_{\tau0}')'$, where  $\omega_{\tau0}=b_{\tau}\omega_0$
and $\vartheta_{\tau0}=b_{\tau}(\alpha_{0+}, \alpha_{0-}, \beta_0)'$. Under some regularity conditions,
the quantile estimator of $\vartheta_{\tau0}$ is shown to be asymptotically normal
for either $\gamma_0<0$ or $\gamma_0\geq 0$, while the quantile estimator of $\omega_{\tau0}$ is asymptotically normal only for
$\gamma_0<0$. Our findings are similar to those in Jensen and Rahbek (2004a, b) and Francq and Zako\"{i}an (2012, 2013a), and
our asymptotic results for $\gamma_0\geq 0$ are the first try of the quantile estimation for non-stationary ARCH-type models in the literature.
Compared to the Gaussian QMLE in Francq and Zako\"{i}an (2013a), our quantile estimator
takes the quantile structure of $\epsilon_t$ into account through the transformation function $T(\cdot)$, and
it could be a more appealing tool to investigate the quantile-based measures such as VaR and ES
(Engle and Manganelli, 2004; Francq and Zako\"{i}an, 2015). Moreover, our quantile estimator only requires
$E|\eta_{t}|^{2r}<\infty$ for its asymptotics, and hence it is more appropriate to study the heavy-tailed financial data than the
Gaussian QMLE, which requires
$E|\eta_{t}|^{4}<\infty$ for its asymptotic normality. As a by-product, new tests for strict stationarity and asymmetry of model (\ref{1.1})
are derived from our estimation procedure.


The remaining of this paper is organized as follows. Section 2 introduces our
hybrid conditional quantile estimation procedure. Section 3
studies the asymptotic properties of our proposed quantile estimator. The strict stationarity tests
and the asymmetry tests are provided in Section 4. Simulation
results are reported in Section 5. Applications are presented in
Section 6. The conclusions are offered in Section 7. The proofs
are given in the Appendix.

Throughout the paper, $|\cdot|$ denotes the absolute value, $\|\cdot\|$ denotes the vector $l_{2}$-norm, $\|\cdot\|_p$ denotes $L^{p}$-norm for a random variable,
$A'$ is the transpose of matrix $A$,
$\to_{p}$ denotes the convergence in probability, $\to_{d}$ denotes the convergence in distribution,
$o_{p}(1)$ (or $O_{p}(1)$) denotes a sequence of random numbers converging to zero ( or bounded) in probability, $C$ is a generic constant, $\mathcal{R}=(-\infty,\infty)$, $\mathcal{R}_{+}=(0,\infty)$,
$\mbox{I}(\cdot)$ is the indicator function,
and $\mbox{sgn}(a)=\mbox{I}(a>0)-\mbox{I}(a<0)$ is the sign of any $a\in\mathcal{R}$.

\section{The hybrid conditional quantile estimation}

Let $\theta=(\omega,\alpha_{+},\alpha_{-},\beta)'\in\Theta$ be the unknown parameter vector of model (1.1), and $\theta_0\in\Theta$
be its
true value, where $\Theta$ is the parameter space, and it is a compact subset of $\mathcal{R}_{+}^{4}$. Moreover, let $\theta_{\tau}=b_{\tau}\theta\in\Theta_{\tau}$, and $\theta_{\tau0}$ be its true value, where
$\Theta_{\tau}=\{\theta_\tau: \theta\in\Theta\}$.
Assume that $\{\epsilon_1,\epsilon_2,...,\epsilon_n\}$ are observations generated from model (\ref{1.1}). By (\ref{1.3}), the parametric $\tau$th quantile of the transformed data $y_t$ is
\begin{flalign}\label{2.1}
Q_{\tau}(y_t|\mathcal{F}_{t-1})=b_{\tau}(\omega+\alpha_{+}(\epsilon_{t-1}^+)^{\delta}+\alpha_{-}(-\epsilon_{t-1}^-)^{\delta}+\beta h_{t-1})=\theta_{\tau}'z_t.
\end{flalign}
If $\{h_{t-1}\}$ are observable, we are able to estimate $Q_{\tau}(y_t|\mathcal{F}_{t-1})$ by the linear quantile regression.
However, $\{h_{t-1}\}$ are not observable, and we shall replace them by some initial estimates. To accomplish this, we define $h_t(\theta)$ recursively by
$$h_t(\theta)=\omega+\alpha_+(\epsilon_{t-1}^+)^{\delta}+\alpha_-(-\epsilon_{t-1}^-)^{\delta}+\beta h_{t-1}(\theta).$$
Then, $h_t=h_t(\theta_0)$. In practice, we calculate
$h_t^{1/\delta}(\theta)$ by $\sigma_t(\theta)$, where
$$\sigma^{\delta}_t(\theta)=\omega+\alpha_+(\epsilon_{t-1}^+)^{\delta}+\alpha_-(-\epsilon_{t-1}^-)^{\delta}
+\beta\sigma^{\delta}_{t-1}(\theta)$$
with given initial values  $\varepsilon_0$ and $\sigma^{\delta}_0(\theta)$.

Based on (\ref{2.1}) and (\ref{1.4}), our hybrid conditional quantile estimation procedure for
$Q_{\tau}(\epsilon_t|\mathcal{F}_{t-1})$ has the following three steps.

{\it Step 1} (Estimation of the global model structure). Using the generalized quasi-maximum likelihood estimator  (GQMLE)
in Francq and Zako\"{i}an (2013b) to estimate the parameter in model (\ref{1.1}),
\begin{flalign}\label{2.2}
\tilde{\theta}_{n,r}=(\tilde{\omega}_{n,r},\tilde{\vartheta}_{n,r}')'&=
\argmin_{\theta\in\Theta}
\frac{1}{n}\sum_{t=1}^{n}\log\left[\sigma_{t}^{r}(\theta)\right]+\frac{|\epsilon_{t}|^{r}}{\sigma_{t}^{r}(\theta)} \nonumber\\
&\equiv
\argmin_{\theta\in\Theta} \frac{1}{n}\sum_{t=1}^{n}l_{t,r}(\theta),
\end{flalign}
where $r$ is a user-chosen positive number.
Based on $\tilde{\theta}_{n,r}$, compute the initial estimates of $\{h_t\}$ as $\{\sigma_t^{\delta}(\tilde{\theta}_{n,r})\}$.

{\it Step 2} (Quantile regression at a specific level). Perform the weighted linear quantile regression of $y_t$ on $\tilde{z}_{t}=(1,(\epsilon_{t-1}^+)^{\delta},(-\epsilon_{t-1}^-)^{\delta},\sigma_{t-1}^{\delta}(\tilde{\theta}_{n,r}))'$ at quantile level $\tau$,
\begin{flalign}\label{2.3}
\hat{\theta}_{\tau n, r}=(\hat{\omega}_{\tau n,r},\hat{\vartheta}_{\tau n,r}')'&=\argmin_{\theta_{\tau}\in\Theta_\tau}\frac{1}{n}\sum_{t=1}^n
\frac{\rho_{\tau}(y_t-\theta_{\tau}'\tilde{z}_{t})}{\sigma_t^{\delta}(\tilde{\theta}_{n,r})}\nonumber\\
&=\argmin_{\theta_{\tau}\in\Theta_\tau}\frac{1}{n}\sum_{t=1}^n
\rho_{\tau}\left(\frac{y_t-\theta_{\tau}'\tilde{z}_{t}}{\sigma_t^{\delta}(\tilde{\theta}_{n,r})}\right)\nonumber\\
&\equiv \argmin_{\theta_{\tau}\in\Theta_\tau}\frac{1}{n}\sum_{t=1}^n l_{t,\rho}(\theta_{\tau}),
\end{flalign}
where $\rho_{\tau}(x)=x[\tau-\mbox{I}(x<0)]$.
Based on $\hat{\theta}_{\tau n, r}$, estimate the $\tau$th conditional quantile of $y_t$
by $\hat{Q}_{\tau}(y_t|\mathcal{F}_{t-1})=\hat{\theta}_{\tau n,r}'\tilde{z}_{t}$.

{\it Step 3} (Transforming back to $\epsilon_t$). Estimate the $\tau$th conditional quantile of the original observation $\epsilon_t$ by $\hat{Q}_{\tau}(\epsilon_t|\mathcal{F}_{t-1})=T^{-1}(\hat{\theta}_{\tau n,r}'\tilde{z}_{t})$.

For the GQMLE $\tilde{\theta}_{n,r}$ in Step 1, Francq and Zako\"{i}an (2013b) established its asymptotic normality
under some regularity conditions. The non-negative user-chosen number $r$ involved in $\tilde{\theta}_{n,r}$ indicates the
estimation method used.
Particularly, when $r=2$,
$\tilde{\theta}_{n,r}$ reduces to the Gaussian QMLE; and when $r=1$,
$\tilde{\theta}_{n,r}$ reduces to the Laplacian QMLE. So far, how to choose an ``optimal'' $r$ (under certain criterion) is unclear, and
simulation studies in Section 5 suggest that we could choose a small (or large) value of $r$ when $\eta_t$ is heavy-tailed (or light-tailed).

For the quantile estimator $\hat{\theta}_{\tau n,r}$ in Step 2, Zheng et al. (2018) studied its asymptotics for a special case that
$\delta=2$ and $\alpha_{0+}=\alpha_{0-}$ with $\gamma_0<0$ (i.e., the stationary classical GARCH model) and $r=2$ (i.e., the Gaussian QMLE).
In the present paper, we will study the asymptotic properties of $\hat{\theta}_{\tau n,r}$ for the general case.

\section{Asymptotic properties of the hybrid quantile estimator}
In this section, we study the asymptotic properties of the hybrid conditional quantile estimator.
First, we give some technical assumptions as follows:

\begin{ass}\label{asm3.1}
(i) $\theta_0$ is an interior point of $\Theta$; (ii) the random variable $\eta_t$ can not concentrate on at most two values, the positive line or the negative line, and $P(|\eta_t|=1)<1$; (iii) $E|\eta_t|^{r}=1$.
\end{ass}

\begin{ass}\label{asm3.2}
The density $f(\cdot)$ of $T(\eta_t)$ is positive and differentiable almost everywhere on $\mathcal{R}$.
\end{ass}

\begin{ass}\label{asm3.3}
When $t$ tends to infinity,
\begin{flalign*} 
E\left\{1+\sum_{i=1}^{t-1}a_0(\eta_1)\ldots a_0(\eta_i)\right\}^{-1}=o\left(\frac{1}{\sqrt{t}}\right).
\end{flalign*}
\end{ass}

Assumption \ref{asm3.1}(i)-(ii) used by Francq and Zako\"{i}an (2013a) are usually assumed for ARCH-type models.
Assumption \ref{asm3.1}(iii) is the identification condition for the GQMLE; see Francq and Zako\"{i}an (2013b).
If $r=\delta$, we have
$$E(|\epsilon_{t}|^{\delta}|\mathcal{F}_{t-1})=h_{t}E|\eta_{t}|^{\delta}=h_t$$
by (\ref{1.1}) and Assumption \ref{asm3.1}(iii), meaning that we can directly predict the $\delta$th moment of $|\epsilon_{t}|$ by
$h_{t}$. If $r\not=\delta$, the $\delta$th moment of $|\epsilon_{t}|$ has to be predicted by $h_{t}E|\eta_{t}|^{\delta}$ in this general case.

Assumption \ref{asm3.2} is standard for quantile estimation.
Assumption \ref{asm3.3} is needed only for $\gamma_0=0$,
and it is used to prove that when $\gamma_0=0$,
$\frac{1}{\sqrt{n}}\sum_{t=1}^{n}\frac{1}{h_t}\to 0$ as $n\to\infty$
in $L^{1}$ (see Francq and Zako\"{i}an, 2012 and 2013a).

Let $\kappa_{1r}=\{E[|\eta_t|^r\mbox{I}(\eta_t<Q_{\tau,\eta})]-\tau\}/r$ and $\kappa_{2r}=(E|\eta_t|^{2r}-1)/r^2$. Define the $4 \times 4$ matrices:
\begin{flalign*}
&J=E\left[\frac{1}{h_t^{2}}\frac{\partial h_t(\theta_0)}{\partial \theta}\frac{\partial h_t(\theta_0)}{\partial \theta'}\right],\,\,\,\quad\Omega=E\left[\frac{z_tz_t'}{h_t^{2}}\right],\\
&H=E\left[\frac{z_{t}}{h_t^{2}}\frac{\partial h_t(\theta_0)}{\partial \theta'}\right],\,\,\,\quad\quad\quad\quad\,\,
\Gamma=E\left[\frac{\beta_0z_{t}}{h_t^{2}}\frac{\partial h_{t-1}(\theta_0)}{\partial \theta'}\right],
\end{flalign*}
and the $3 \times 3$ matrices:
\begin{flalign*}
&J_{\vartheta}=E[d_t(\vartheta_0)d_t(\vartheta_0)'],\ \ \ \ \ \Omega_{\vartheta}=E[\xi_{t}\xi_{t}'],\\
&H_{\vartheta}=E\left[{\xi}_{t}d_t(\vartheta_0)'\right], \ \ \ \ \ \ \ \  \Gamma_{\vartheta}=E\left[\beta_0{\xi}_{t}\frac{d_{t-1}(\vartheta_0)'}{a_0(\eta_{t-1})}\right],
\end{flalign*}
where $d_t(\vartheta)$ is defined in (\ref{a.1}), and
$$\xi_{t}=\left(\f{(\eta_{t-1}^+)^{\delta}}{a_0(\eta_{t-1})},\frac{(-\eta_{t-1}^-)^{\delta}}{a_0(\eta_{t-1})},\frac{1}{a_0(\eta_{t-1})}\right)'.$$

\begin{thm}\label{thm3.1}
Suppose that Assumptions \ref{asm3.1}-\ref{asm3.2} hold and $E|\eta_t|^{2r}<\infty$.

(i) [Stationary case] When $\gamma_0<0$,  and $\beta<1$ for all $\theta\in \Theta$,
\begin{flalign}\label{3.1}
\sqrt{n}(\hat{\theta}_{\tau n,r}-\theta_{\tau 0})\rightarrow_d N(0,\Sigma_{r}),
\end{flalign}
where
\begin{flalign*}
\Sigma_{r}=\Omega^{-1}\left[\frac{\tau-\tau^2}{f^2(b_{\tau})}\Omega+\frac{\kappa_{1r}\delta b_{\tau}}{f(b_{\tau})}(\Gamma J^{-1}H'+HJ^{-1}\Gamma')+\kappa_{2r}\delta^2b_{\tau}^2\Gamma J^{-1}\Gamma'\right]\Omega^{-1}.
\end{flalign*}

(ii) [Explosive case]  When $\gamma_0>0$, and $P(\eta_t=0)=0$,
\begin{flalign}\label{3.2}
\sqrt{n}(\hat{\vartheta}_{\tau n,r}-\vartheta_{\tau 0})\rightarrow _d N(0,\Sigma_{\vartheta,r}),
\end{flalign}
where
\begin{flalign*}
\Sigma_{\vartheta,r}=\Omega_{\vartheta}^{-1}\left[\frac{\tau-\tau^2}{f^2(b_{\tau})}\Omega_{\vartheta}+\frac{\kappa_{1r}\delta b_{\tau}}{ f(b_{\tau})}\left(\Gamma_{\vartheta}J_{\vartheta}^{-1}H_{\vartheta}'+H_{\vartheta}J_{\vartheta}^{-1}\Gamma_{\vartheta}'\right)+\kappa_{2r}\delta^2 b_{\tau}^2\Gamma_{\vartheta}J_{\vartheta}^{-1}\Gamma_{\vartheta}'\right]\Omega_{\vartheta}^{-1}.
\end{flalign*}

(iii) [At the boundary of the stationarity region] When $\gamma_0=0$, $P(\eta_t=0)=0$, $\beta<\|1/a_0(\eta_t)\|_p^{-1}$
for any $\theta\in \Theta$ and some $p>1$, and Assumption \ref{asm3.3} is satisfied,
then (\ref{3.2}) holds.
\end{thm}

\begin{rem}
Similar to the Gaussian QMLE in Jensen and Rahbek (2004a, b) and Francq and Zako\"{i}an (2012, 2013a), $\hat{\vartheta}_{\tau n,r}$ is always asymptotically normal distributed regardless of the sign of $\gamma_0$, and  $\hat{\omega}_{\tau n,r}$ is shown to be asymptotically normal distributed only for
$\gamma_0<0$.

Our results in Theorem \ref{thm3.1} are also related to those in Zheng et al. (2018), but with three major differences.
First, the results in Theorem 1 of Zheng et al. (2018) are nested by ours with $\gamma_0<0$, $\alpha_{0+}=\alpha_{0-}$ and
$\delta=r=2$. Second, the results in Zheng et al. (2018) need the assumption $E|\eta_t|^4<\infty$,
while our results hold under a weaker assumption $E|\eta_t|^{2r}<\infty$, which is applicable to the heavy-tailed $\eta_t$.
Third, the results of Zheng et al. (2018) are only for the stationary GARCH model, but our results cover both stationary and non-stationary asymmetric PGARCH models,
leading to a much larger applicability scope than theirs.
\end{rem}

\begin{rem}
To prove the result in (iii), a technical condition $\beta<\|1/a_0(\eta_t)\|_p^{-1}$ is needed, and it poses an additional restriction on the parameter $\beta$.
Clearly, the boundary point $\|1/a_0(\eta_t)\|_p^{-1}$ is related to the constant $p$, the distribution of $\eta_t$, and the value of $(\delta, \alpha_{0+}, \alpha_{0-}, \beta_0)$.
Table \ref{uppbound} reports the values of $\|1/a_0(\eta_t)\|_p^{-1}$ for several choices of $p$, $\eta_t$, and $\delta$, where the value of $\beta_0$ is fixed to be 0.9, the value of $\alpha_{0-}$ is set to be $0.01, 0.04, ..., 0.25$, and the value of $\alpha_{0+}$ is uniquely determined by the condition $\gamma_0=0$.
From this table, we can find that (i) the value of $\beta_0$ always lies in the region $\{\beta:\beta<\|1/a_0(\eta_t)\|_p^{-1}\}$; (ii)
the values of $\|1/a_0(\eta_t)\|_p^{-1}$ do not vary too much across $\alpha_{0-}$ or the distribution of $\eta_t$, although they become slightly smaller as the values of $p$
become larger. In sum, based on our calculations, the technical condition $\beta<\|1/a_0(\eta_t)\|_p^{-1}$ seems mild, and it should not
hinder the practical application of our proposed estimation.
\begin{table}
\caption{\label{uppbound}The  values of $\|1/a_0(\eta_t)\|_p^{-1}$ when $\gamma_0=0$ with $\beta_0=0.9$}
\small\addtolength{\tabcolsep}{-1.8pt}
\renewcommand{\arraystretch}{0.84}
\begin{tabular}{cccccccccccccccc}

\hline


$\eta_t$&\diagbox[width=3em]{$p$}{$\alpha_{0-}$}&0.01&0.04&0.07&0.10&0.13&0.16&0.19&0.22&0.25\\
\hline\\

\multicolumn{11}{c}{Panel A: $\delta=2$}\\

$N(0, 1)$&2 &0.97366& 0.98019&0.98380& 0.98524& 0.98497& 0.98325& 0.98023&0.97599& 0.97066\\
&4   & 0.95886& 0.96792& 0.97274& 0.97465& 0.97429& 0.97201& 0.96797&0.96215& 0.95448\\
&6&0.94949&0.95953& 0.96467& 0.96667&0.96630& 0.96391& 0.95958&0.95320&0.94441\\\\

$t_5$&2 &0.96867& 0.97439& 0.97750& 0.97894& 0.97913& 0.97831& 0.97662&0.97410& 0.97075\\
&4&0.95403& 0.96143& 0.96531& 0.96708& 0.96732& 0.96631& 0.96421&0.96106& 0.95677\\
&6&0.94528& 0.95323& 0.95727& 0.95909& 0.95934& 0.95831& 0.95614&0.95284& 0.94826 \\\\

$t_2$&2 &0.96093&0.96276& 0.95718& 0.96282& 0.97221& 0.98027& 0.98736&0.99368& 0.99940\\
&4& 0.94825& 0.95038& 0.94380& 0.94596& 0.95183& 0.95670& 0.96087& 0.96450& 0.96772\\
&6& 0.94116& 0.94335& 0.93651& 0.93704&0.94125& 0.94468& 0.94756&0.95003& 0.95219 \\\\

\multicolumn{11}{c}{Panel B: $\delta=1$}\\

$N(0, 1)$&2 & 0.98360& 0.98868& 0.99209& 0.99401& 0.99459 &0.99397 &0.99224& 0.98952 &0.98587\\
&4&  0.97119 &0.97972 &0.98545& 0.98867& 0.98964 &0.98859& 0.98570 &0.98113 &0.97501\\
&6& 0.96174 &0.97257& 0.97982 &0.98389& 0.98512& 0.98379 &0.98013 &0.97435& 0.96659\\\\

$t_5$&2 &0.98177 &0.98659 &0.98993 &0.99198& 0.99290 &0.99279& 0.99176& 0.98987 &0.98720\\
&4&0.96894 &0.97679 &0.98217 &0.98547 &0.98694 &0.98676 &0.98511 &0.98208 &0.97776 \\
&6&0.95955 &0.96931 &0.97597 &0.98002& 0.98182 &0.98161 &0.97958 &0.97585 &0.97052\\\\

$t_2$&2 & 0.96174 &0.97257 &0.97982& 0.98389 &0.98512& 0.98379 &0.98013 &0.97435& 0.96659\\
&4& 0.96629& 0.97588& 0.97941& 0.97865 &0.97438 &0.96686 &0.96342& 0.96892& 0.97385\\
&6&0.95788 &0.96930& 0.97347& 0.97258& 0.96753 &0.95856& 0.95315& 0.95703& 0.96043\\
\hline
\end{tabular}
\end{table}
\end{rem}

\begin{rem}
Our results in Theorem \ref{thm3.1} are derived for a known exponent $\delta$. When $\delta$ is unknown in general, we can
include $\delta$ as an additional unknown parameter in our first estimation procedure, and the asymptotics of the resulting GQMLE can be established with
some minor modifications (see also Section 6 in Francq and Zako\"{i}an (2013a)). However, since the unknown exponent $\delta$ is involved in the transformation function $T(\cdot)$, how to derive the asymptotics of the corresponding quantile estimator in the second step estimation procedure is challenging at this stage, and we leave this interesting topic for the future study.
\end{rem}

Let $\bar{z}_{t,\vartheta}=((\epsilon_{t-1}^+)^{\delta},(-\epsilon_{t-1}^-)^{\delta},\sigma^{\delta}_{t-1}(\theta_0))'.$ By (\ref{a.22})-(\ref{a.23}) and Lemma \ref{lema3}, we have
\begin{flalign}
\sqrt{n}(\hat{\theta}_{\tau n,r}-\theta_{\tau 0})&=\Omega^{-1}
\left[\frac{1}{\sqrt{n}}\sum_{t=1}^{n}\left(Uu_{t}
+Vv_{t}\right)\right]+o_{p}(1) \nonumber\\
&\equiv \Omega^{-1}
\left[\frac{1}{\sqrt{n}}\sum_{t=1}^{n}e_{t}\right]+o_{p}(1), \label{a3.3}\\
\sqrt{n}(\hat{\vartheta}_{\tau n,r}-\vartheta_{\tau 0})&=\Omega_{\vartheta}^{-1}
\left[\frac{1}{\sqrt{n}}\sum_{t=1}^{n}\left(Uu_{\vartheta,t}
+V_{\vartheta}v_{\vartheta,t}\right)\right]+o_{p}(1) \nonumber\\
&\equiv \Omega_{\vartheta}^{-1}
\left[\frac{1}{\sqrt{n}}\sum_{t=1}^{n}e_{\vartheta,t}\right]+o_{p}(1), \label{a3.4}
\end{flalign}
where $U=1/f(b_{\tau})$  and
\begin{flalign*}
&V=\frac{b_{\tau}\delta}{r}\Gamma J^{-1},\,\,
u_{t}=\psi_{\tau}\left(\eta_t-Q_{\tau,\eta}\right)
\f{z_{t}}{h_t(\theta_{0})}, \,\,
v_{t}=[1-|\eta_t|^r]\frac{1}{h_{t}}\frac{\partial h_{t}(\theta_0)}{\partial\theta}, \\
&V_{\vartheta}=\frac{b_{\tau}\delta}{r}\Gamma_{\vartheta}J_{\vartheta}^{-1},\,\,
u_{\vartheta,t}=\psi_{\tau}\left(\eta_t-Q_{\tau,\eta}\right)
\f{\bar{z}_{t,\vartheta}}{\sigma^{\delta}_t(\theta_{0})}, \,\,
v_{\vartheta,t}=[1-|\eta_t|^r]\frac{1}{h_{t}}\frac{\partial \sigma_{t}^{\delta}(\theta_0)}{\partial\vartheta}
\end{flalign*}
with $\psi_{\tau}(x)=\tau-\mathrm{I}(x<0)$.

Based on $\tilde{\theta}_{n,r}$, we can calculate $\tilde{\Omega}_{r}$, $\tilde{U}_{r}$, $\tilde{u}_{r,t}$, $\tilde{b}_{\tau,r}$,
$\tilde{\Gamma}_{r}$, $\tilde{J}_{r}$, and $\tilde{v}_{r,t}$, which are
the sample counterparts of
$\Omega$, $U$, $u_t$, $b_{\tau}$, $\Gamma$, $J$, and $v_{t}$, respectively\footnote{For $\tilde{U}_{r}$,  we follow Silverman (1986) to
estimate $f(x_0)$ by the Gaussian kernel density estimator  $\tilde{f}(x_0)={\sum_{t=1}^n K_h(T(\tilde{\eta}_{t,r})-x_0)}/{n}$ with $K_h(x)={1}/{(\sqrt{2\pi}h)}\exp\{-x^2/(2h^2)\}$
and the rule-of-thumb bandwidth $h=0.9n^{-1/5}\min(s, \tilde{R}/1.34)$, where $\tilde{\eta}_{t,r}=\epsilon_t/\sigma_t(\tilde{\theta}_{n,r})$, and
$s$ and $\tilde{R}$ are the sample standard deviation and interquartile range of the transformed residuals $\{T(\tilde{\eta}_{t,r})\}$, respectively.}.
Since $e_t$ is a martingale difference sequence,
by (\ref{a3.3}) we can estimate
$\Sigma_{r}$ by
\begin{flalign*}
\tilde{\Sigma}_{r}=\tilde{\Omega}_{r}^{-1}\left[\frac{1}{n}\sum_{t=1}^{n}\tilde{e}_{r,t}\tilde{e}_{r,t}'\right]\tilde{\Omega}_{r}^{-1},
\end{flalign*}
where $\tilde{e}_{r,t}=\tilde{U}_{r}\tilde{u}_{r,t}
+\tilde{V}_{r}\tilde{v}_{r,t}$ with $\tilde{V}_{r}=(\tilde{b}_{\tau,r}\delta/r)\tilde{\Gamma}_{r} \tilde{J}_{r}^{-1}$. Under the conditions of Theorem \ref{thm3.1}(i), we can show that
$\tilde{\Sigma}_{r}$ is a consistent estimator of $\Sigma_{r}$ for $\gamma_0<0$.

Partition $\tilde{u}_{r,t}=(\tilde{u}_{\omega r,t},\tilde{u}_{\vartheta r,t}')'$, $\tilde{v}_{r,t}=(\tilde{v}_{\omega r,t},\tilde{v}_{\vartheta r,t}')'$, and
$$\tilde{\Sigma}_{r}=
\begin{bmatrix}
\tilde{\Sigma}_{\omega\omega, r} & \tilde{\Sigma}_{\omega\vartheta, r} \\
\tilde{\Sigma}_{\omega\vartheta, r}' & \tilde{\Sigma}_{\vartheta\vartheta, r}
\end{bmatrix},\,\,
\tilde{\Omega}_{r}=
\begin{bmatrix}
\tilde{\Omega}_{\omega\omega, r} & \tilde{\Omega}_{\omega\vartheta, r} \\
\tilde{\Omega}_{\omega\vartheta, r}' & \tilde{\Omega}_{\vartheta\vartheta, r}
\end{bmatrix},\,\,
\tilde{\Gamma}_{r}=
\begin{bmatrix}
\tilde{\Gamma}_{\omega\omega, r} & \tilde{\Gamma}_{\omega\vartheta, r} \\
\tilde{\Gamma}_{\omega\vartheta, r}' & \tilde{\Gamma}_{\vartheta\vartheta, r}
\end{bmatrix},\,\,
\tilde{J}_{r}=
\begin{bmatrix}
\tilde{J}_{\omega\omega, r} & \tilde{J}_{\omega\vartheta, r} \\
\tilde{J}_{\omega\vartheta, r}' & \tilde{J}_{\vartheta\vartheta, r}
\end{bmatrix}.$$
Then, $\tilde{\Omega}_{\vartheta\vartheta, r}$, $\tilde{u}_{\vartheta r, t}$, $\tilde{\Gamma}_{\vartheta\vartheta, r}$, $\tilde{J}_{\vartheta\vartheta, r}$ and
$\tilde{v}_{\vartheta r, t}$ are the sample counterparts of
$\Omega_{\vartheta}$, $u_{\vartheta, t}$, $\Gamma_{\vartheta}$, $J_{\vartheta}$ and
$v_{\vartheta, t}$, respectively. Since $e_{\vartheta, t}$ is a martingale difference sequence,
by (\ref{a3.4}) we can estimate
$\Sigma_{\vartheta, r}$ by
\begin{flalign*}
\tilde{\Sigma}_{\vartheta,r}=\tilde{\Omega}_{\vartheta\vartheta, r}^{-1}\left[\frac{1}{n}\sum_{t=1}^{n}\tilde{e}_{\vartheta r,t}\tilde{e}_{\vartheta r,t}'\right]\tilde{\Omega}_{\vartheta\vartheta, r}^{-1},
\end{flalign*}
where $\tilde{e}_{\vartheta r,t}=\tilde{U}_{r}\tilde{u}_{\vartheta r,t}
+\tilde{V}_{\vartheta, r}\tilde{v}_{\vartheta r,t}$ with $\tilde{V}_{\vartheta, r}=(\tilde{b}_{\tau,r}\delta/r)\tilde{\Gamma}_{\vartheta\vartheta, r} \tilde{J}_{\vartheta\vartheta, r}^{-1}$. Under the conditions of Theorem \ref{thm3.1}(ii)-(iii), we can show that
$\tilde{\Sigma}_{\vartheta,r}=\Sigma_{\vartheta, r}+o_{p}(1)$ and $\tilde{\Sigma}_{\vartheta\vartheta, r}=\tilde{\Sigma}_{\vartheta,r}+o_{p}(1)$ for $\gamma_0\geq 0$, which implies that we can estimate $\Sigma_{\vartheta, r}$ by $\tilde{\Sigma}_{\vartheta\vartheta, r}$ for either $\gamma_0<0$ or $\gamma_0\geq 0$.

\section{Strict stationarity and asymmetry tests}
\subsection{Testing for strict stationarity}
Since the stationarity of model (\ref{1.1}) is determined by the sign of $\gamma_0$, it is interesting
to consider the strict stationarity testing problems as follows:
\begin{flalign}\label{4.1}
H_0: \gamma_0<0\ \ \mbox{against}\ H_1: \gamma_0\geq 0,
\end{flalign}
and
\begin{flalign}\label{4.2}
H_0: \gamma_0\geq0\ \ \mbox{against}\  H_1: \gamma_0< 0.
\end{flalign}
In Francq and Zako\"{i}an (2013a), a strict stationarity test based on the Gaussian QMLE is proposed. In this subsection, similar to Francq and Zako\"{i}an (2013a), we construct a  strict stationarity test based on the GQMLE.

For any $\theta\in \Theta$, let $\eta_t(\theta)=\epsilon_t/\sigma_t(\theta)$ and
\begin{flalign*}
\gamma_n(\theta)=\f{1}{n}\sum_{t=1}^n\log[\alpha_+(\eta_t^+(\theta))^{\delta}+\alpha_-(-\eta_t^-(\theta))^{\delta}+\beta].
\end{flalign*}
Then,  we can estimate $\gamma_0$ by $\tilde{\gamma}_{n,r}=\gamma_n(\tilde{\theta}_{n,r})$.
The following result shows the asymptotic distribution of $\tilde{\gamma}_{n,r}$ in both stationary and nonstationary cases.

\begin{cor}\label{cor4.1}
Let $u_t=\log(a_0(\eta_t))-\gamma_0$, $\sigma_u^2=E(u_t^2)$ and $a=(0,E\xi_{t}')'$. Then, under the conditions of Theorem \ref{thm3.1},
\begin{flalign} \label{4.4}
\sqrt{n}(\tilde{\gamma}_{n,r}-\gamma_0)\rightarrow_d N(0,\sigma_{\gamma_0}^2)\ \ \mbox{as}\ \  n \rightarrow \infty,
\end{flalign}
where
\begin{flalign*}
\sigma_{\gamma_0}^2=
\left\{\begin{array}{ll}
\sigma_u^2+\delta^2\kappa_{2r} \{a'J^{-1}a-(1-E[\f{\beta_0}{a_0(\eta_t)}])^2\},&\ \ \mbox{as}\ \ \gamma_0<0,\\
\sigma_{u}^2, &  \ \ \mbox{as}\ \ \gamma_0\geq0.
\end{array}
\right.
\end{flalign*}
\end{cor}

The proof of Corollary \ref{cor4.1} is omitted, since it is similar to the one in Francq and Zako\"{i}an (2013a)
except for some minor modifications.
Let $\tilde{\eta}_{t,r}=\eta_{t}(\tilde{\theta}_{n,r})$.
Under the conditions of Corollary \ref{cor4.1}, $\sigma_u^2$ can be consistently estimated by
$\tilde{\sigma}_{u,r}^2$, where $\tilde{\sigma}_{u,r}^2$ is the sample
variance of $\{\log[\tilde{\alpha}_{n+,r}(\tilde{\eta}_{t,r}^+)^{\delta}+\tilde{\alpha}_{n-,r}(-\tilde{\eta}_{t,r}^-)^{\delta}+\tilde{\beta}_{n,r}]\}$. Then, the statistic
$$\hat{T}_{r}=\sqrt{n}\tilde{\gamma}_{n,r}/\tilde{\sigma}_{u,r}$$
asymptotically converges to $N(0,1)$  when $\gamma_0=0$.
For the testing  problem (\ref{4.1}) [or (\ref{4.2})], this leads us to consider the critical region
\begin{flalign}\label{4.5}
C^{ST}=\{\hat{T}_{r}>\Phi^{-1}(1-\underline{\alpha})\} \ \ [\mbox{or }C^{NT}=\{\hat{T}_{r}<\Phi^{-1}(\underline{\alpha})\}]
\end{flalign}
at the asymptotic significance level of $\underline{\alpha}$.

\subsection{Testing for asymmetry}
Testing for the existence of asymmetry (or leverage) effect is important in many financial applications.
For model (\ref{1.1}),  this asymmetry testing problem is of the form
\begin{flalign}\label{4.6}
H_0: \alpha_{0+}=\alpha_{0-}\ \ \mbox{against}\ H_1: \alpha_{0+}\neq\alpha_{0-}.
\end{flalign}
In this subsection, we propose two tests for the hypotheses in (\ref{4.6}).
Let $\tilde{\sigma}_{S,r}^{*}=\sqrt{e'\tilde{\Sigma}_{\vartheta\vartheta,r}^{*}e}$ and
$\tilde{\sigma}_{S,r}=\sqrt{e'\tilde{\Sigma}_{\vartheta\vartheta,r}e}$ with $e=(1,-1,0)'$,
where $\tilde{\Sigma}_{\vartheta\vartheta,r}$ defined before is a consistent estimator of the
asymptotic variance of $\hat{\vartheta}_{\tau n,r}$, and
$$\tilde{\Sigma}_{\vartheta\vartheta,r}^{*}=\frac{\delta^2}{r^2}\tilde{J}_{\vartheta\vartheta, r}^{-1}
\left[\frac{1}{n}\sum_{t=1}^{n}\tilde{v}_{\vartheta r,t}\tilde{v}_{\vartheta r,t}'\right]\tilde{J}_{\vartheta\vartheta, r}^{-1}.$$
By Lemmas \ref{lema1}-\ref{lema4} and the similar argument as for Theorem 3.2 in Francq and Zako\"{i}an (2013a), we can show that
$\tilde{\Sigma}_{\vartheta\vartheta,r}^{*}$ is a consistent estimator of the
asymptotic variance of $\tilde{\vartheta}_{n,r}$. With  $\tilde{\sigma}_{S,r}^{*}$ and $\tilde{\sigma}_{S,r}$,
our test statistics for asymmetry are defined by
\begin{flalign*}
\hat{S}_{1,r}=\f{\sqrt{n}(\tilde{\alpha}_{n+,r}-\tilde{\alpha}_{n-,r})}{\tilde{\sigma}_{S,r}^{*}}
\,\,\mbox{ and }\,\,\hat{S}_{2,r}^{(\tau)}=\f{\sqrt{n}(\hat{\alpha}_{\tau n+,r}-\hat{\alpha}_{\tau n-,r})}{\tilde{\sigma}_{S,r}}.
\end{flalign*}
Note that $\hat{S}_{1,r}$ is based on the GQMLE, and it aims to examine the asymmetric effect in model (\ref{1.1}) globally, while
$\hat{S}_{2,r}^{(\tau)}$ does this locally at a specific quantile level $\tau$ by using
the quantile estimator. Under the conditions of Theorem \ref{thm3.1}, it is straightforward that
both $\hat{S}_{1,r}$ and $\hat{S}_{2,r}^{(\tau)}$ asymptotically converge to $N(0,1)$ under $H_0$ in (\ref{4.6}). Hence,
 the critical region based on $\hat{S}_{1,r}$ [or $\hat{S}_{2,r}^{(\tau)}$] is
\begin{flalign}\label{4.7}
C^S=\{|\hat{S}_{1,r}|>\Phi^{-1}(1-\underline{\alpha}/2)\} \ \ [\mbox{or, }C^S=\{|\hat{S}_{2,r}^{(\tau)}|>\Phi^{-1}(1-\underline{\alpha}/2)\}]
\end{flalign}
for the testing problem (\ref{4.6}), and it has the asymptotic significance level $\underline{\alpha}$.
Since $\tilde{\alpha}_{n\pm,r}$, $\hat{\alpha}_{\tau n\pm,r}$, $\tilde{\sigma}_{S,r}^{*}$ or $\tilde{\sigma}_{S,r}$ has the unified asymptotics for both
$\gamma_0<0$ and $\gamma_0\geq0$, the tests $\hat{S}_{1,r}$ and $\hat{S}_{2,r}^{(\tau)}$ can be used in both cases. This is also the situation for the asymmetry test in Francq and Zako\"{i}an (2013a).
We shall emphasize that unlike the Gaussian QMLE-based tests in Francq and Zako\"{i}an (2013a), our tests
$\hat{T}_{r}$, $\hat{S}_{1,r}$ and $\hat{S}_{2,r}^{(\tau)}$ only require $E|\eta_t|^{2r}<\infty$, and they thus
are valid for the very heavy-tailed $\eta_t$.

\section{Simulation studies}
\subsection{Simulation studies for the quantile estimators}
In this section, we assess the finite-sample
performance of $\hat{\theta}_{\tau n,r}$. We generate 1000 replications
from the following model:
\begin{flalign}\label{5.1}
\epsilon_t=h_{t}^{1/\delta}\eta_t, \quad h_t=0.1+\alpha_{0+}(\epsilon_{t-1}^+)^{\delta}+0.15(-\epsilon_{t-1}^-)^{\delta}+0.9h_{t-1},
\end{flalign}
where $\eta_t$ is taken as $N(0, 1)$, the standardized Student's
$t_5$ ($\mathrm{st}_5$) or the standardized Student's
$t_3$ ($\mathrm{st}_3$) such that $E\eta_{t}^2=1$. Here, we
fix $\omega_0=0.1$, $\alpha_{0-}=0.15$ and $\beta_0=0.9$, and choose $\alpha_{0+}$ as in Table\,\ref{alpha_value}, where
the values of $\alpha_0+$ correspond to
the cases of $\gamma_{0}>0$, $\gamma_{0}=0$, and $\gamma_{0}<0$,
respectively. For the power index $\delta$ (or the estimation indicator $r$), we choose it to be 2 or 1. For the quantile level $\tau$, we set it
to be $0.05$ or $0.1$.
Since
each GQMLE has a different identification condition, $\hat{\theta}_{\tau n,r}$
has to be re-scaled for $\theta_{\tau0}$ in model (\ref{5.1}), and it is defined as
$$\hat{\theta}_{\tau n,r}=\left(\bar{\omega}_{\tau n,r}, \bar{\alpha}_{\tau n+,r}, \bar{\alpha}_{\tau n-,r},
(E|\eta_{t}|^{r})^{\delta/r}\bar{\beta}_{\tau n,r}\right),$$
where $\overline{\theta}_{n,r}=(\bar{\omega}_{\tau n,r},\bar{\alpha}_{\tau n+,r},\bar{\alpha}_{\tau n-,r},\bar{\beta}_{\tau n,r})'$ is the hybrid quantile estimator calculated from the data sample, and the true value of
$(E|\eta_{t}|^{r})^{\delta/r}$ is used.

\begin{table}[!htbp]\scriptsize
\caption{\label{alpha_value}The values of the pair $(\alpha_{0+}, \gamma_{0})$ when $\alpha_{0-}=0.15, \beta_{0}=0.9$}
\addtolength{\tabcolsep}{-5.5pt}
 \renewcommand{\arraystretch}{1.1}
\begin{tabular}{cccccccccccccccccc}


\hline
\multicolumn{8}{c}{$\delta=2$} & &
\multicolumn{8}{c}{$\delta=1$}\\
\cline{1-8}\cline{10-17}
\multicolumn{2}{c}{$\eta_{t}\sim N(0, 1)$} & &
\multicolumn{2}{c}{$\eta_{t}\sim \mathrm{st}_{5}$} & &
\multicolumn{2}{c}{$\eta_{t}\sim \mathrm{st}_{3}$} & &
\multicolumn{2}{c}{$\eta_{t}\sim N(0, 1)$} & &
\multicolumn{2}{c}{$\eta_{t}\sim \mathrm{st}_{5}$} & &
\multicolumn{2}{c}{$\eta_{t}\sim \mathrm{st}_{3}$}\\

\cline{1-2}\cline{4-5}\cline{7-8}\cline{10-11}\cline{13-14}\cline{16-17}

\multicolumn{1}{c}{$\alpha_{0+}$} &\multicolumn{1}{c}{$\gamma_{0}$} &
& \multicolumn{1}{c}{$\alpha_{0+}$} & \multicolumn{1}{c}{$\gamma_{0}$} &
& \multicolumn{1}{c}{$\alpha_{0+}$} & \multicolumn{1}{c}{$\gamma_{0}$} &
&\multicolumn{1}{c}{$\alpha_{0+}$} &\multicolumn{1}{c}{$\gamma_{0}$} &
& \multicolumn{1}{c}{$\alpha_{0+}$} & \multicolumn{1}{c}{$\gamma_{0}$} &
& \multicolumn{1}{c}{$\alpha_{0+}$} & \multicolumn{1}{c}{$\gamma_{0}$}\\

\hline

0.05  & -0.0104& & 0.05 & -0.0152 & & 0.05  &-0.0226 && 0.05 & -0.0233 && 0.05 & -0.0261 && 0.05 &-0.0286\\

0.07224697 &0.0000 & &0.09206513  & 0.0000 & &  0.1516561 &0.0000 && 0.1083685 & 0.0000 && 0.1332366&0.0000
&& 0.1830638 &0.0000\\

0.2  &0.0517 & & 0.2  &0.0330 & & 0.2  &0.0091 & &0.2  &0.0337 & & 0.2  &0.0192 & & 0.2  &0.0034\\

\hline

\end{tabular}
\end{table}

Tables \ref{gqmle2}  and \ref{gqmle3} report the
bias, the empirical standard deviation (ESD) and the asymptotic standard deviation (ASD) of $\hat{\theta}_{\tau n,r}$ for the cases of $\delta=2$ and $\delta=1$, respectively. In this section, since the results for $\eta_{t}\sim \mathrm{st}_3$ are similar, they are not reported here for saving space. From Tables \ref{gqmle2}  and \ref{gqmle3}, our findings are as follows:

\begin{sidewaystable}
\caption{\label{gqmle2}Summary for $\hat{\theta}_{\tau n,r}$ ($\times 10$) when $\delta=2$}
\small\addtolength{\tabcolsep}{-5pt}
\begin{tabular}{lllcccccccccccccccccccc ccccc ccccc}
\hline

& &
& \multicolumn{9}{c}{$\gamma_0<0$} &
& \multicolumn{9}{c}{$\gamma_0=0$} &
& \multicolumn{9}{c}{$\gamma_0>0$} \\

\cline{4-12}\cline{14-22}\cline{24-32}

& &
& \multicolumn{4}{c}{$r=2$} &
& \multicolumn{4}{c}{$r=1$} &
& \multicolumn{4}{c}{$r=2$} &
& \multicolumn{4}{c}{$r=1$} &
& \multicolumn{4}{c}{$r=2$} &
& \multicolumn{4}{c}{$r=1$}\\

\cline{4-7}\cline{9-12}\cline{14-17}\cline{19-22}\cline{24-27}\cline{29-32}

\multicolumn{1}{c}{$\eta_{t}$} & \multicolumn{1}{c}{$n$} &
& $\omega$ & $\alpha_{+}$ & $\alpha_{-}$ & $\beta$ &
& $\omega$ & $\alpha_{+}$ & $\alpha_{-}$ & $\beta$ &
& $\omega$ & $\alpha_{+}$ & $\alpha_{-}$ & $\beta$ &
& $\omega$ & $\alpha_{+}$ & $\alpha_{-}$ & $\beta$ &
& $\omega$ & $\alpha_{+}$ & $\alpha_{-}$ & $\beta$ &
& $\omega$ & $\alpha_{+}$ & $\alpha_{-}$ & $\beta$ \\

\hline\\

\multicolumn{32}{c}{Panel A: $\tau=0.05$}\\

$N(0, 1)$&1000&Bias& -1.36& -0.68& -0.36& 0.62&&-1.29 &-0.55&-0.28&0.40&&-1.17&-0.51&-0.37&0.36&& -0.72&-0.52&-0.34&0.48&&-1.74 &-0.41&-0.24&0.22&&-0.93&-0.16&-0.32&0.16\\
&&ESD&5.63 & 1.81& 2.72&2.33 &&6.20 &1.82&2.88&2.49&&5.45&1.99&2.74&2.23&&5.66&2.11&2.88&2.42&&6.33&3.14&2.65&2.35&&6.10&3.13&2.80&2.45\\
&&ASD& 3.99& 2.11& 2.78& 2.43 &&4.30 &2.14&2.80&2.54&&3.88&2.25&2.75&2.34&&4.26&2.32&2.84&2.45&&5.34&3.09&2.77&2.41&&4.76&3.18&2.85&2.53\\

&2000&Bias&-1.08&-0.37&-0.12&0.29&&-1.35&-0.24&-0.14&0.25&&-1.32&-0.24&-0.13&0.14&&-0.66&-0.21&-0.05&-0.15&&-1.65&-0.22&-0.11&-0.07&&-1.41&-0.09&-0.10&0.03\\
&&ESD&4.76&1.31&2.01&1.69&&5.33&1.35&2.08&1.81&&5.96&1.48&1.95&1.63&&4.94&1.55&1.98&1.67&&6.01&2.27&2.01&1.70&&6.47&2.20&2.01&1.69\\
&&ASD&3.97&1.56&2.02&1.76&&4.24&1.55&2.04&1.83&&4.55&1.67&2.03&1.70&&4.40&1.65&1.99&1.72&&5.05&2.27&2.02&1.73&&5.21&2.28&2.04&1.79\\

$\mbox{st}_{5}$&1000&Bias&-1.72&-0.84&-0.70&0.60&&-1.57&-0.88&-0.35&0.53&&-3.28&-0.85&-0.82&0.17&&-8.57&-0.26&-0.56&0.11&&-4.70&-0.75&-0.63&0.08&&-3.32&-0.43&-0.58&0.15\\
&&ESD&7.72&2.26&3.53&3.18&&6.92&2.43&3.40&2.78&&10.9&2.92&3.34&2.72&&4.77&8.98&3.52&2.96&&21.9&3.96&3.18&2.89&&18.6&4.12&3.31&2.64\\
&&ASD&5.75&2.18&3.34&3.15&&5.05&2.20&3.17&2.89&&6.51&2.69&3.23&2.84&&15.9&8.34&3.19&2.91&&7.35&3.60&3.18&2.83&&6.47&3.55&3.17&2.63\\

&2000&Bias&-1.35&-0.46&-0.33&0.39&&-1.71&-0.40&-0.18&0.35&&-3.12&-0.57&-0.35&0.12&&-1.99&-0.39&-0.19&0.15&&-5.37&-0.56&-0.44&0.07&&-3.24&-0.35&-0.29&0.04\\
&&ESD&6.03&1.51&2.40&2.27&&5.79&1.64&2.37&2.05&&13.2&2.05&2.27&1.96&&10.7&2.06&2.28&1.81&&28.1&2.81&2.51&2.03&&13.8&2.83&2.36&1.82\\
&&ASD&5.56&1.59&2.36&2.29&&4.75&1.59&2.34&2.11&&6.11&1.97&2.35&2.02&&5.19&1.95&2.36&1.89&&7.54&2.70&2.35&2.04&&6.74&2.70&2.39&1.93\\\\

%
%

\multicolumn{32}{c}{Panel B: $\tau=0.1$}\\

$N(0, 1)$&1000&Bias& -0.84&-0.45&-0.25 &0.36&&-0.62 &-0.38&-0.18&0.32 && -1.10&-0.47&-0.35&0.28&&-0.87&-0.40&-0.25&0.26&&-1.18&-0.25&-0.18&0.12&&-0.86&-0.17&-0.09&-0.10\\
&&ESD&3.42& 1.15&1.67&1.56&&3.29 &1.17&1.69&1.57&&3.71&1.31&1.72&1.40&&3.68&1.33&1.75&1.41&&3.79&1.97&1.71&1.53&&3.69&2.00&1.73&1.53\\
&&ASD&3.05&1.38&1.80&1.58&&3.03 &1.38&1.80&1.58&&3.32&1.51&1.82&1.54&&3.35&1.50&1.81&1.54&&3.91&2.00&1.79&1.55&&3.95&2.00&1.79&1.55\\

&2000&Bias&-0.85&-0.25&-0.16&0.28&&-0.64&-0.21&-0.11&0.25&&-0.87&-0.16&-0.14&0.09&&-0.62&-0.12&-0.10&0.07&&-1.34&-0.14&-0.08&0.07&&-0.97&-0.11&-0.04&0.06\\
&&ESD&3.10&0.87&1.27&1.08&&3.04&0.87&1.27&1.07&&3.36&0.95&1.30&1.05&&3.19&0.95&1.30&1.06&&3.96&1.39&1.21&1.09&&3.86&1.41&1.22&1.11\\
&&ASD&2.83&0.10&1.30&1.13&&2.83&0.99&1.30&1.13&&3.05&1.06&1.29&1.08&&3.04&1.06&1.29&1.08&&3.94&1.42&1.28&1.10&&3.94&1.42&1.28&1.10\\

$\mbox{st}_{5}$&1000&Bias&-1.62&-0.43&-0.29&0.31&&-1.36&-0.33&-0.16&0.34&&-1.65&-0.50&-0.36&0.06&&-1.32&-0.39&-0.22&0.12&&-2.66&-0.34&-0.39&-0.05&&-1.96&-0.24&-0.23&0.04\\
&&ESD&4.42&1.03&1.55&1.68&&4.38&1.00&1.54&1.52&&4.57&1.41&1.56&1.47&&4.60&1.38&1.55&1.31&&8.28&1.95&1.67&1.56&&7.47&1.93&1.63&1.40\\
&&ASD&3.21&1.15&1.67&1.62&&2.95&1.13&1.65&1.47&&3.56&1.41&1.67&1.49&&3.31&1.40&1.66&1.35&&4.27&1.90&1.67&1.51&&4.14&1.88&1.64&1.35\\

&2000&Bias&-1.09&-0.23&-0.20&0.19&&-0.85&-0.18&-0.11&0.19&&-1.72&-0.19&-0.17&0.05&&-1.20&-0.12&-0.21&0.06&&-2.86&-0.22&-0.19&0.03&&2.13&-0.03&-0.11&-0.01\\
&&ESD&3.05&0.75&1.23&1.19&&3.02&0.79&1.26&1.05&&6.24&0.98&1.17&1.07&&4.19&0.94&1.19&0.98&&8.59&1.35&1.14&1.06&&7.76&1.33&1.19&0.10\\
&&ASD&2.94&0.82&1.22&1.20&&2.70&0.82&1.21&1.09&&3.70&0.10&1.21&1.07&&3.29&0.99&1.21&0.96&&4.25&1.37&1.20&1.08&&4.39&1.35&1.20&0.97\\

%
%


\hline
\end{tabular}
\end{sidewaystable}

\begin{sidewaystable}
\caption{\label{gqmle3}Summary for $\hat{\theta}_{\tau n,r}$ ($\times 10$) when $\delta=1$}
\small\addtolength{\tabcolsep}{-5pt}
\begin{tabular}{lllcccccccccccccccccccc ccccc ccccc}

\hline

& &
& \multicolumn{9}{c}{$\gamma_0<0$} &
& \multicolumn{9}{c}{$\gamma_0=0$} &
& \multicolumn{9}{c}{$\gamma_0>0$} \\

\cline{4-12}\cline{14-22}\cline{24-32}

& &
& \multicolumn{4}{c}{$r=2$} &
& \multicolumn{4}{c}{$r=1$} &
& \multicolumn{4}{c}{$r=2$} &
& \multicolumn{4}{c}{$r=1$} &
& \multicolumn{4}{c}{$r=2$} &
& \multicolumn{4}{c}{$r=1$}\\

\cline{4-7}\cline{9-12}\cline{14-17}\cline{19-22}\cline{24-27}\cline{29-32}

\multicolumn{1}{c}{$\eta_{t}$} & \multicolumn{1}{c}{$n$} &
& $\omega$ & $\alpha_{+}$ & $\alpha_{-}$ & $\beta$ &
& $\omega$ & $\alpha_{+}$ & $\alpha_{-}$ & $\beta$ &
& $\omega$ & $\alpha_{+}$ & $\alpha_{-}$ & $\beta$ &
& $\omega$ & $\alpha_{+}$ & $\alpha_{-}$ & $\beta$ &
& $\omega$ & $\alpha_{+}$ & $\alpha_{-}$ & $\beta$ &
& $\omega$ & $\alpha_{+}$ & $\alpha_{-}$ & $\beta$ \\

\hline\\

\multicolumn{32}{c}{Panel A: $\tau=0.05$}\\

$N(0, 1)$&1000&Bias&-0.32&-0.67&-0.56&0.55& &-0.11&-0.59&-0.44&0.45&&-0.38&-0.55&-0.51&0.41&&-0.31&-0.44&-0.39&0.35&&-0.85&-0.40&-0.41&0.29&&-0.57&-0.25&-0.23&0.23\\

&&ESD&1.53&0.93&1.23&0.96&&1.51&0.93&1.25&0.95&&1.53&1.08&1.29&0.87&&1.63&1.12&1.24&0.91&&1.87&1.38&1.26&0.95&&1.93&1.39&1.22&0.96\\

&&ASD&1.28&1.23&1.36&1.00&&1.33&1.22&1.36&1.02&&1.49&1.30&1.35&0.97&&1.60&1.29&1.36&0.99&&2.01&1.43&1.36&1.00&&1.97&1.41&1.35&1.01\\

&2000&Bias&-0.46&-0.40&-0.32&-0.39&&-0.28&-0.39&-0.29&0.34&&-0.45&-0.26&-0.21&0.21&&-0.37&-0.20&-0.19&0.16&&-0.75&-0.20&-0.20&0.16&&-0.50&-0.12&-0.16&0.14\\

&&ESD&1.76&0.69&0.91&0.78&&1.65&0.74&0.94&0.79&&1.52&0.84&0.89&0.64&&1.66&0.82&0.93&0.68&&1.94&1.04&0.93&0.69&&1.76&1.02&0.93&0.72\\

&&ASD&1.28&0.88&0.98&0.76&&1.29&0.88&0.98&0.77&&1.51&0.93&0.97&0.70&&1.50&0.93&0.97&0.71&&1.82&1.02&0.97&0.71&&1.84&1.01&0.97&0.72\\

$\mbox{st}_{5}$&1000&Bias&-0.69&-0.87&-0.66&0.62&&-0.47&-0.64&-0.50&0.51&&-1.14&-0.60&-0.61&0.31&&-0.74&-0.56&-0.55&0.34&&-1.23&-0.55&-0.52&0.23&&-1.00&-0.24&-0.40&0.21\\
&&ESD&2.34&1.14&1.55&1.32&&2.46&1.20&1.56&1.34&&2.76&1.51&1.54&1.09&&2.79&1.57&1.53&1.07&&2.82&1.66&1.56&1.08&&3.05&1.65&1.48&1.09\\
&&ASD&1.71&1.46&1.67&1.24&&1.73&1.43&1.66&1.22&&2.09&1.61&1.65&1.16&&2.06&1.61&1.65&1.12&&2.53&1.78&1.66&1.18&&2.27&1.71&1.63&1.12\\

&2000&Bias&-0.81&-0.49&-0.38&0.48&&-0.61&-0.40&-0.29&0.42&&-0.97&-0.28&-0.34&0.16&&-0.83&-0.22&-0.24&0.17&&-1.37&-0.23&-0.27&0.12&&-0.88&-0.23&-0.26&0.14\\
&&ESD&2.24&0.84&1.14&1.11&&2.25&0.86&1.14&1.03&&2.42&1.09&1.14&0.82&&2.79&1.13&1.19&0.77&&2.85&1.21&1.11&0.82&&2.85&1.25&1.16&0.81\\
&&ASD&1.70&1.04&1.20&0.97&&1.68&1.04&1.19&0.94&&2.17&1.16&1.19&0.84&&2.10&1.16&1.19&0.80&&2.64&1.26&1.19&0.84&&2.25&1.26&1.19&0.81\\\\

%
%

\multicolumn{32}{c}{Panel B: $\tau=0.1$}\\

$N(0, 1)$&1000&Bias&-0.17&-0.48&-0.43&0.37&&-0.07&-0.40&-0.29&0.28&&-0.33&-0.37&-0.32&0.30&&-0.17&0.25&-0.24&0.18&&-0.59&-0.29&-0.29&0.20&&-0.45&-0.18&-0.25&0.19\\

&&ESD&1.31&0.78&1.05&0.81&&1.26&0.74&1.05&0.79&&1.24&0.92&1.00&0.71&&1.20&0.90&1.00&0.72&&1.40&1.07&1.04&0.76&&1.47&1.12&1.01&0.79\\

&&ASD&1.28&1.02&1.14&0.86&&1.23&1.02&1.13&0.85&&1.44&1.07&1.13&0.81&&1.42&1.08&1.13&0.82&&1.82&1.19&1.13&0.83&&1.77&1.18&1.13&0.83\\

&2000&Bias&-0.26&-0.26&-0.21&0.24&&-0.11&-0.22&-0.21&0.17&&-0.40&-0.20&-0.22&0.16&&-0.33&-0.14&-0.11&0.11&&-0.58&-0.12&-0.11&0.08&&-0.41&-0.09&-0.07&0.06\\

&&ESD&1.30&0.57&0.76&0.63&&1.19&0.57&0.78&0.62&&1.27&0.69&0.74&0.53&&1.27&0.70&0.72&0.54&&1.43&0.82&0.76&0.57&&1.43&0.81&0.76&0.57\\

&&ASD&1.25&0.73&0.81&0.65&&1.22&0.73&0.81&0.65&&1.52&0.77&0.81&0.58&&1.43&0.77&0.80&0.58&&1.83&0.84&0.81&0.59&&1.74&0.84&0.81&0.59\\

$\mbox{st}_{5}$&1000&Bias&-0.44&-0.49&-0.46&0.36&&-0.21&-0.42&-0.30&0.31&&-0.57&-0.41&-0.41&0.24&&-0.51&-0.27&-0.30&0.19&&-0.87&-0.35&-0.38&0.16&&-0.72&-0.19&-0.27&0.13\\
&&ESD&1.59&0.78&1.09&0.96&&1.40&0.81&1.07&0.86&&1.61&1.03&1.09&0.80&&1.62&1.01&1.06&0.72&&1.81&1.15&1.07&0.80&&2.18&1.17&1.08&0.76\\
&&ASD&1.43&1.02&1.18&0.92&&1.34&1.01&1.17&0.86&&1.66&1.14&1.16&0.83&&1.63&1.12&1.16&0.78&&2.03&1.23&1.17&0.84&&1.86&1.22&1.16&0.79\\

&2000&Bias&-0.43&-0.24&-0.20&0.23&&-0.33&-0.27&-0.21&0.23&&-0.63&-0.17&-0.20&0.08&&-0.48&-0.16&-0.15&0.10&&-0.79&-0.19&-0.17&0.07&&-0.52&-0.13&-0.11&0.07\\
&&ESD&1.45&0.57&0.81&0.75&&1.48&0.63&0.84&0.72&&1.67&0.75&0.82&0.58&&1.53&0.79&0.79&0.53&&1.69&0.88&0.75&0.60&&1.65&0.83&0.81&0.53\\
&&ASD&1.43&0.73&0.84&0.74&&1.30&0.73&0.83&0.68&&1.64&0.81&0.83&0.59&&1.61&0.81&0.83&0.56&&1.96&0.88&0.83&0.60&&1.81&0.88&0.83&0.56\\

%
%

\hline
\end{tabular}
\end{sidewaystable}

  (a1) The biases of all parameters become small as the sample size $n$ increases, except when $\gamma_0\geq 0$, the estimators of $\omega$ have relatively large biases as expected.
  For each distribution of $\eta_t$, the biases of $\hat{\theta}_{\tau n,r}$ with $r=1$ (or $\tau=0.1$) are generally smaller than those of $\hat{\theta}_{\tau n,r}$ with $r=2$  (or $\tau=0.05$). For each estimator, its biases (in absolute value) in the case of $\eta_{t}\sim\mathrm{st}_5$ tend to be smaller than those in the case of $\eta_{t}\sim N(0, 1)$.

  (a2) The ESDs and ASDs of the parameter $\vartheta$ are close in all cases, while
  the ESDs and ASDs of the parameter $\omega$ have a relatively large disparity as expected. As the sample size
  $n$ increases, the ESDs and ASDs of all parameters become small. For each distribution of $\eta_t$,
  the ASDs of $\hat{\theta}_{\tau n,r}$ seem robust to the choices of $r$,
  and they become large as the value of $\tau$ decreases. For each estimator, its ASDs
  in the case of $\eta_{t}\sim\mathrm{st}_5$ are generally larger than those in the case of $\eta_{t}\sim N(0, 1)$, except for
  $\delta=2$ and $\tau=0.1$.

Note that all of the aforementioned findings are invariant, regardless of the power index $\delta$ and the sign of $\gamma_0$.
In summary, our quantile estimator $\hat{\theta}_{\tau n,r}$ has a good finite sample performance, which is robust to the choice of
$r$. Particulary, its performance tends to be even better, when $\eta_t$ is more light-tailed or the value of $\tau$ is larger.

\subsection{Simulation studies for the tests}
In this subsection, we first assess the performance of the strict stationarity test $\hat{T}_{r}$.
We generate 1000 replications from model (\ref{5.1}) with the same settings for $\delta$ and $\eta_t$, except that the values of $\alpha_{0+}$ are chosen as in Table \ref{test1}. We apply $\hat{T}_{r}$ with $r=2$ and $1$ to both testing problems (\ref{4.1}) and (\ref{4.2}) at the significance level of
5\%, and obtain the following findings:

(b1) The size of $\hat{T}_{r}$ is controlled by the level of 5\% in general, though there is some over-sized risk for the
testing problem (\ref{4.2}) when the sample size $n$ is not large enough. This is also observed in Francq and Zako\"{i}an (2012, 2013a).

(b2) The power of  $\hat{T}_{r}$ is satisfactory, and it increases with the sample size $n$.
Also, $\hat{T}_{r}$ is more powerful when the tail of $\eta_{t}$ is thinner.
But the choice of
$r$ has a negligible effect on the power of $\hat{T}_{r}$. This may be because the asymptotic variance of $\tilde{\gamma}_{n,r}$ in (\ref{4.4})
does not depend on $r$.

Next, we assess the performance of asymmetry tests $\hat{S}_{1,r}$ and  $\hat{S}_{2,r}^{(\tau)}$.
As before, we generate 1000 replications from model (\ref{5.1}) with the same settings for $\delta$ and $\eta_t$, except that the values of $\alpha_{0+}$ are chosen to be $\{0.01, 0.03, \cdots, 0.27, 0.29\}$.
We apply $\hat{S}_{1,r}$ and  $\hat{S}_{2,r}^{(\tau)}$ (with $\tau=0.05$ and $0.1$)
to the testing problem (\ref{4.6})
at the significance level of
5\%. Figs.\,\ref{figure1} and \ref{figure2} plot the power of
$\hat{S}_{1,r}$ and  $\hat{S}_{2,r}^{(\tau)}$ for $r=1$ with $\eta_{t}\sim N(0, 1)$ and $\mbox{st}_{5}$, respectively.
Since the results for $r=2$ are similar, we do not show them here for saving the space. Our findings are as follows:

(c1) All three tests have precise sizes even when $n$ is not large.

(c2) The power of all three tests increases when the value of $\alpha_{0+}$ moves away from $0.15$, and the global test $\hat{S}_{1,r}$ is more powerful than the two local tests  $\hat{S}_{2,r}^{(\tau)}$. Both local tests $\hat{S}_{2,r}^{(\tau)}$ are more powerful for $\delta=1$ than for $\delta=2$.
When $\eta_t\sim N(0, 1)$,  $\hat{S}_{2,r}^{(\tau)}$ with $\tau=0.05$ is more powerful than  $\hat{S}_{2,r}^{(\tau)}$ with $\tau=0.1$, while
when $\eta_t\sim \mbox{st}_{5}$, the opposite conclusion is obtained.

Overall, all our proposed tests have a good performance especially for large $n$.

\begin{table}
\caption{\label{test1}Power ($\times$100) of $\hat{T}_{r}$ at the significance level $5\%$}
\small\addtolength{\tabcolsep}{-0.5pt}
\renewcommand{\arraystretch}{0.85}
\begin{tabular}{ccccccccccccc}

\hline

\multicolumn{12}{c}{Panel A: $\delta=2$}\\

&&&&&\multicolumn{7}{c}{$\alpha_{0+}$}\\
\cline{6-12}

$\eta_t$&$H_0$&$r$&&$n$ &0.01&0.03&0.05&0.07224697&0.09&0.11&0.13\\
\hline

$N(0, 1)$&$(\ref{4.1})$&2 &&1000&0.0&0.0&0.0 & \textbf{7.6}&53.7 &96.8&99.8\\
&&&&2000&0.0&0.0&0.0 & \textbf{6.3}& 80.4&100&100\\
&&&&4000&0.0&0.0&0.0 & \textbf{5.8}&97.3 &100&100\\

&&1&&1000&0.0&0.0&0.0 & \textbf{6.4}& 54.0&96.2&100\\
&&&&2000&0.0&0.0& 0.0& \textbf{5.4}&79.3 &100&100\\
&&&&4000&0.0&0.0&0.0 & \textbf{5.0}& 96.9&100&100\\

&$(\ref{4.2})$&2&&1000&100&99.3&78.1&\textbf{14.1}&0.6 & 0.0&0.6 \\
&&&&2000&100&100&93.7&\textbf{11.5}&9.4 &0.0 &0.0 \\
&&&&4000&100&100&99.8&\textbf{10.2}&0.0 &0.0 &0.0 \\

&&1&&1000&100&98.5&77.3&\textbf{16.7}&0.5 & 0.0&0.0 \\
&&&&2000&100&100&93.7&\textbf{13.8}& 0.0&0.0 &0.0 \\
&&&&4000&100&100&99.6&\textbf{8.4}&0.0 &0.0 &0.0 \\\\

&&&&&\multicolumn{7}{c}{$\alpha_{0+}$}\\
\cline{6-12}
$\eta_t$&$H_0$&$r$&&$n$&0.03&0.05&0.07&0.09206513&0.11&0.13&0.15\\
\hline

$\mbox{st}_5$&$(\ref{4.1})$&2 &&1000&0.0&0.0 &0.5 & \textbf{6.3}&35.0 &76.9&96.3\\
&&&&2000&0.0&0.0 &0.0 & \textbf{5.6} &52.9&95.1&99.9\\
&&&&4000&0.0& 0.0&0.0& \textbf{5.3} &78.2&99.0&100\\

&&1&&1000&0.0& 0.0&0.1& \textbf{6.3} &34.6&74.5&95.4\\
&&&&2000&0.0&0.0 &0.0& \textbf{5.8} &54.8&95.3&100\\
&&&&4000&0.0& 0.0&0.0& \textbf{5.1} &73.5&99.8&100\\

&$(\ref{4.2})$&2&&1000&98.8&90.7&58.5& \textbf{17.9} & 3.6&0.5 &0.0 \\
&&&&2000&100&98.3&75.4& \textbf{13.7} &0.8 &0.0 &0.0 \\
&&&&4000&100&100&92.3&\textbf{12.7} &0.1 &0.0 &0.0 \\

&&1&&1000&99.6&99.3&60.9&\textbf{16.7}&1.9 &0.1 &0.0 \\
&&&&2000&100&99.5&79.1&\textbf{13.9} &0.4 &0.0 &0.0 \\
&&&&4000&100&99.9&94.3&\textbf{10.0} &0.0 &0.0 &0.0 \\\\

\multicolumn{12}{c}{Panel B: $\delta=1$}\\

&&&&&\multicolumn{7}{c}{$\alpha_{0+}$}\\
\cline{6-12}

$\eta_t$&$H_0$&$r$&&$n$ &0.05&0.07&0.09&0.1083685&0.13&0.15&0.17\\
\hline

$N(0, 1)$&$(\ref{4.1})$&2 &&1000&0.0&0.0 &0.0& \textbf{6.8}&94.1&100&100\\
&&&&2000&0.0 &0.0&0.0& \textbf{5.5} &99.6&100&100\\
&&&&4000& 0.0& 0.0&0.0 &\textbf{4.8}&100&100&100\\

&&1&&1000&0.0 &0.0&0.0& \textbf{7.2}&93.8 &100&100\\
&&&&2000&0.0 &0.0&0.0& \textbf{5.8}&99.8 &100&100\\
&&&&4000&0.0 &0.0&0.0& \textbf{5.1}& 100&100&100\\

&$(\ref{4.2})$&2&&1000&100&99.9&89.5&\textbf{10.8}& 0.1&0.0&0.0\\
&&&&2000&100&100&99.0&\textbf{9.9} &0.0 & 0.0 &0.0\\
&&&&4000&100&100&100 &\textbf{7.7} &0.0 &0.0&0.0\\

&&1&&1000&100&100&90.5&\textbf{11.9}&0.0 &0.0&0.0\\
&&&&2000&100&100&99.2&\textbf{10.2}&0.0 & 0.0&0.0\\
&&&&4000&100&100&100 & \textbf{7.9}&0.0 &0.0&0.0\\\\

&&&&&\multicolumn{7}{c}{$\alpha_{0+}$}\\
\cline{6-12}
$\eta_t$&$H_0$&$r$&&$n$&0.07&0.09&0.11&0.1332366&0.15&0.17&0.19\\
\hline

$\mbox{st}_5$&$(\ref{4.1})$&2 &&1000&0.0&0.0&0.0& \textbf{8.3}&62.9&98.6&100\\
&&&&2000&0.0&0.0&0.0& \textbf{7.4}& 84.4&100&100\\
&&&&4000&0.0&0.0&0.0& \textbf{5.6}&99.0 &100&100\\

&&1&&1000&0.0&0.0&0.0 & \textbf{7.4}&63.5 &98.8&100\\
&&&&2000& 0.0 &0.0&0.0&\textbf{6.3}&86.8&100&100\\
&&&&4000&0.0 &0.0&0.0& \textbf{4.5}&99.2&100&100\\

&$(\ref{4.2})$&2&&1000&99.9&99.5&83.9 &\textbf{12.6}  &0.4 &0.0&0.0\\
&&&&2000&100&100&97.9&\textbf{11.1}& 0.0 &0.0&0.0\\
&&&&4000&100&100&100&\textbf{10.1}&0.0&0.0&0.0  \\

&&1&&1000&100&99.7 &88.0&\textbf{14.9}&  0.4&0.0&0.0\\
&&&&2000&100&100 &98.7&\textbf{11.8}& 0.0 &0.0&0.0\\
&&&&4000&100&100& 99.9&\textbf{9.0}&0.0 &0.0& 0.0 \\

\hline
\end{tabular}
  \begin{tablenotes}
    \item {\small $\dag$ The size of $\hat{T}_{r}$  is in boldface.}
 \end{tablenotes}
\end{table}

\begin{figure}[!htbp]
\begin{center}
\includegraphics[width=10pc,height=10pc]{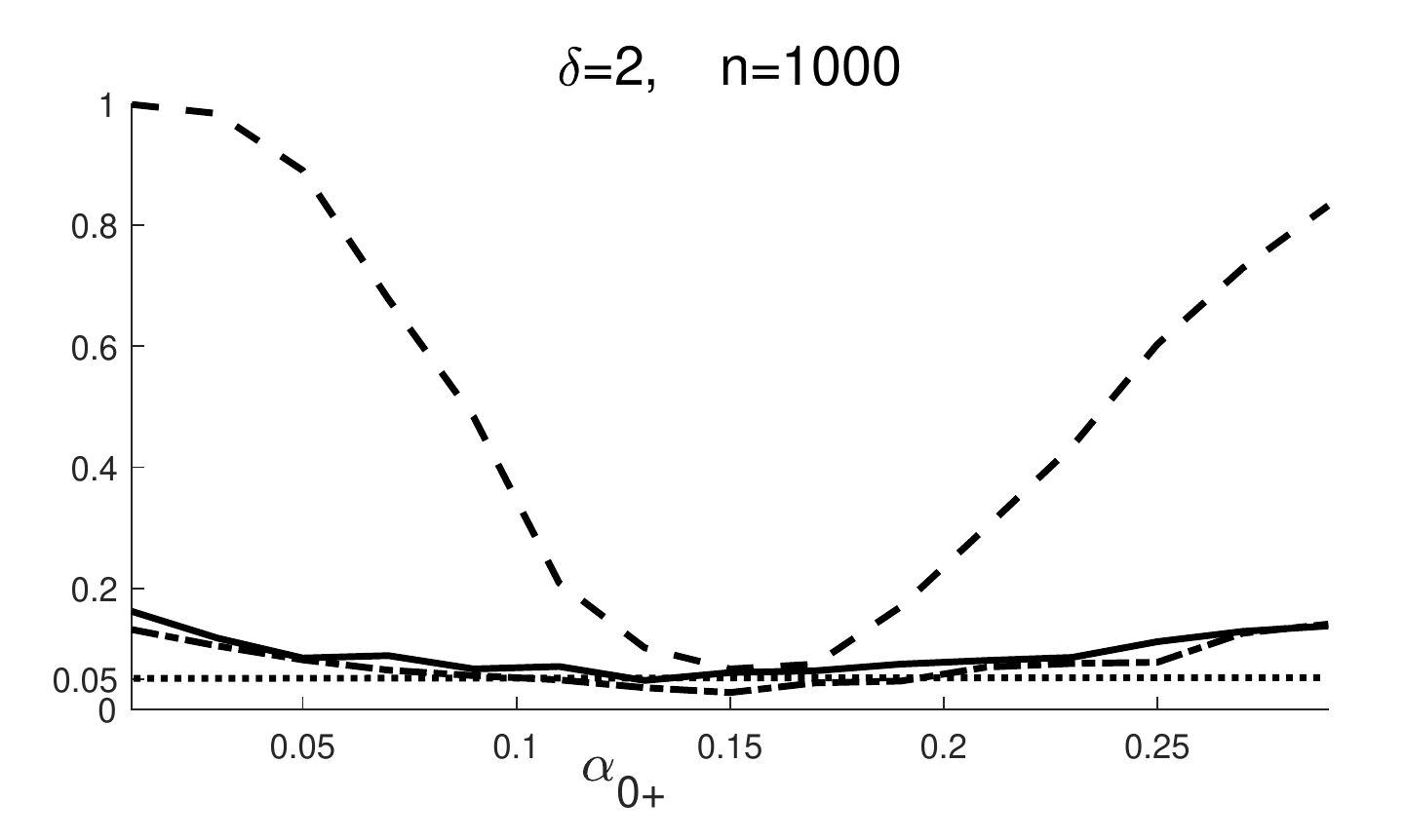}\includegraphics[width=10pc,height=10pc]{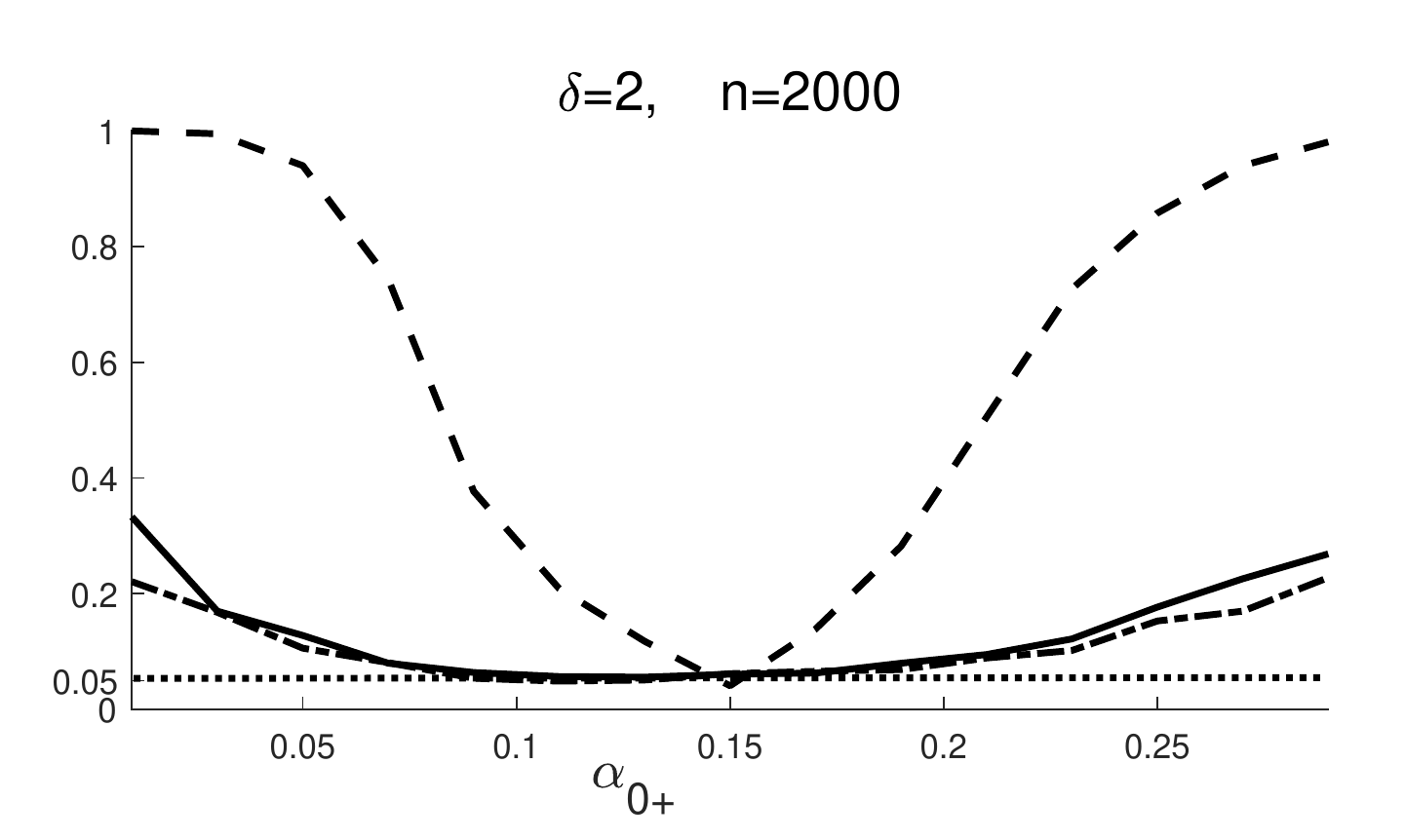}\includegraphics[width=10pc,height=10pc]{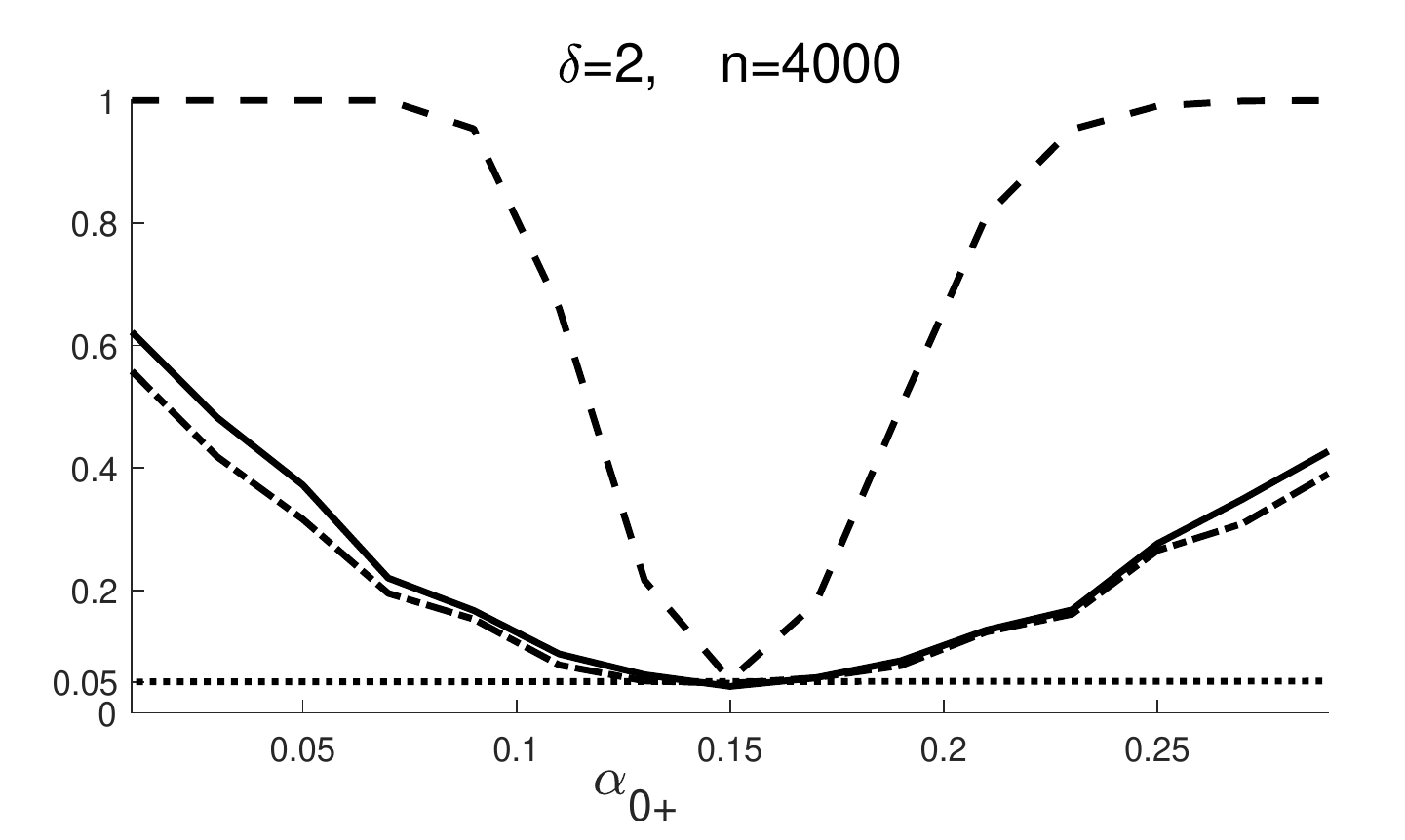}\\
\includegraphics[width=10pc,height=10pc]{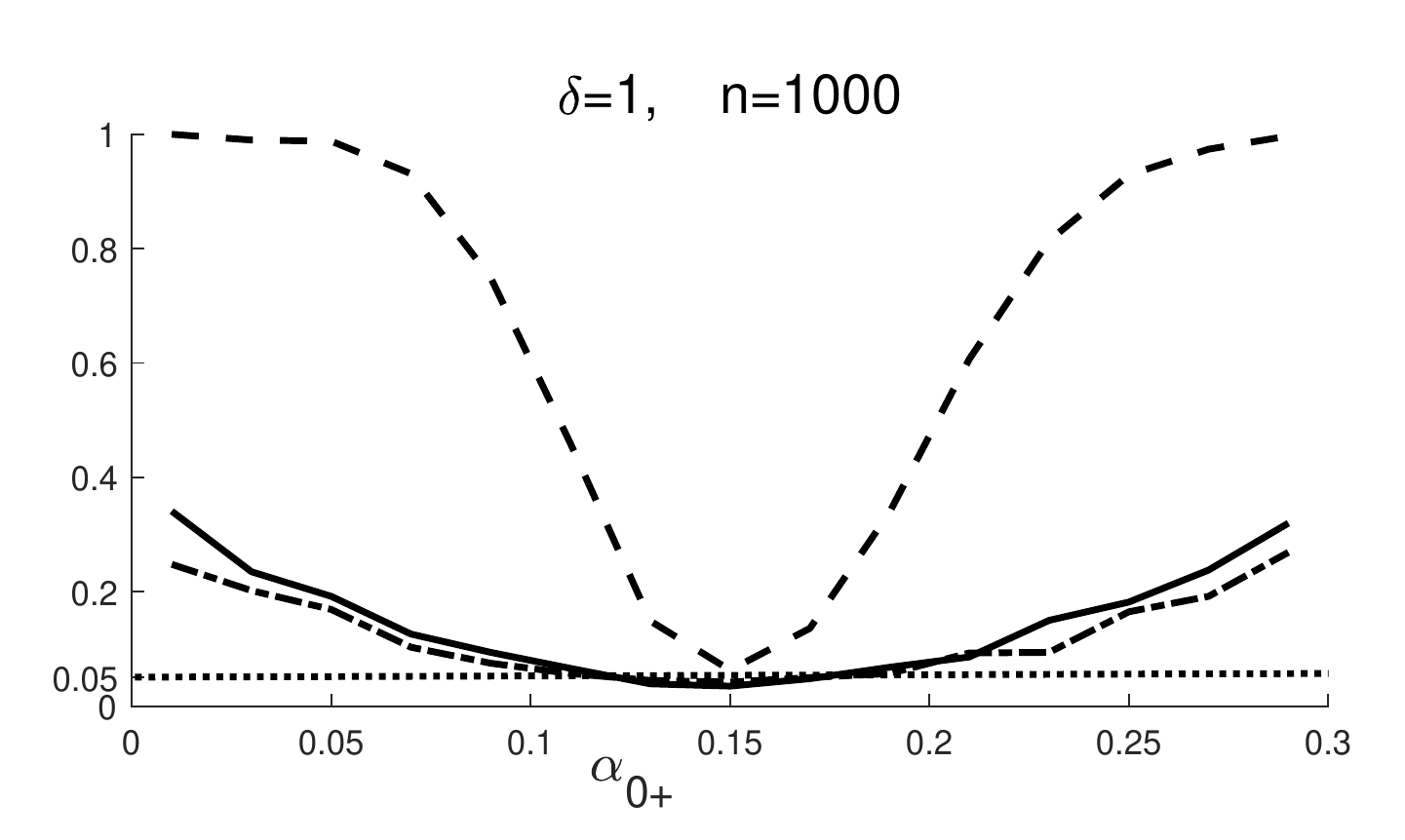}\includegraphics[width=10pc,height=10pc]{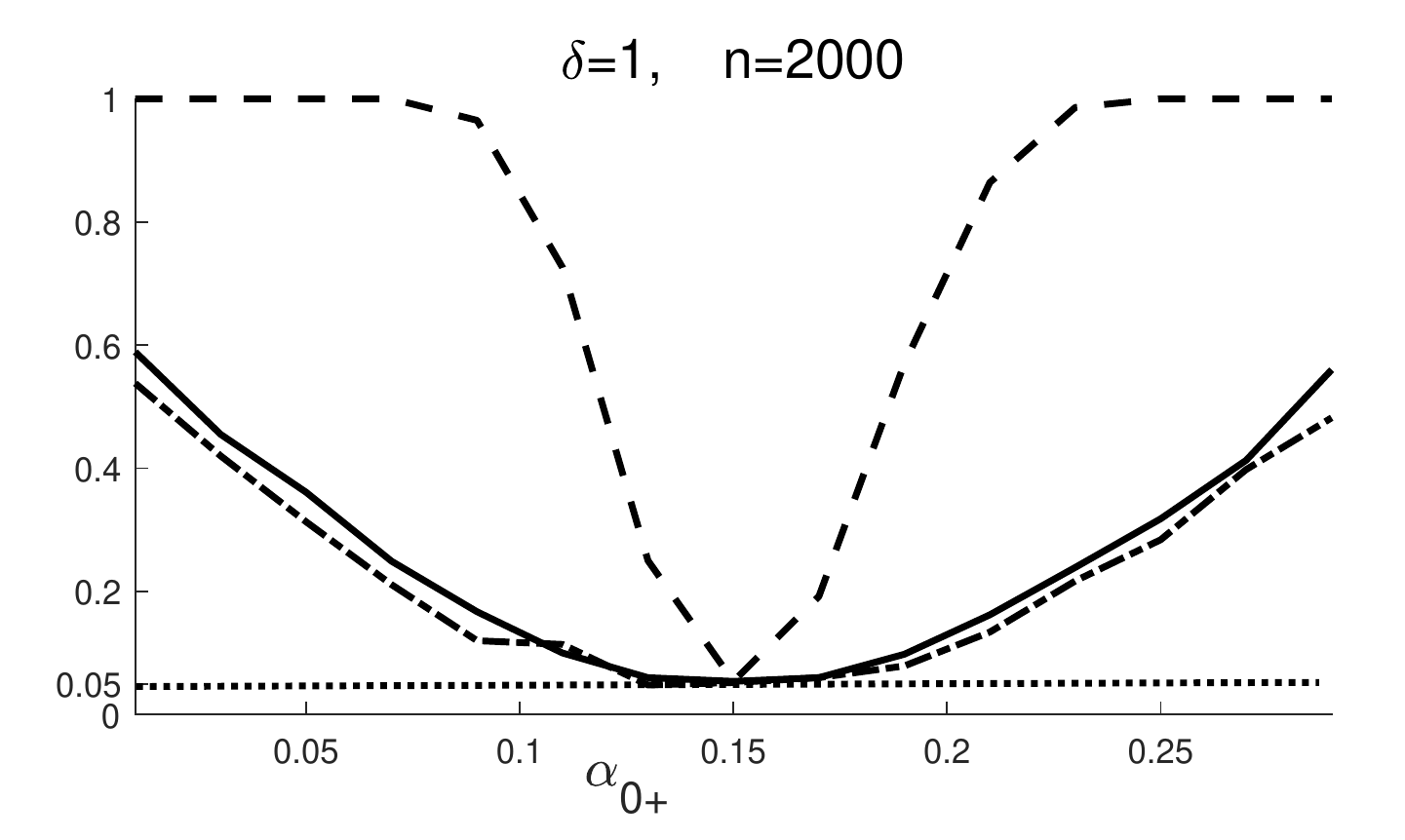}\includegraphics[width=10pc,height=10pc]{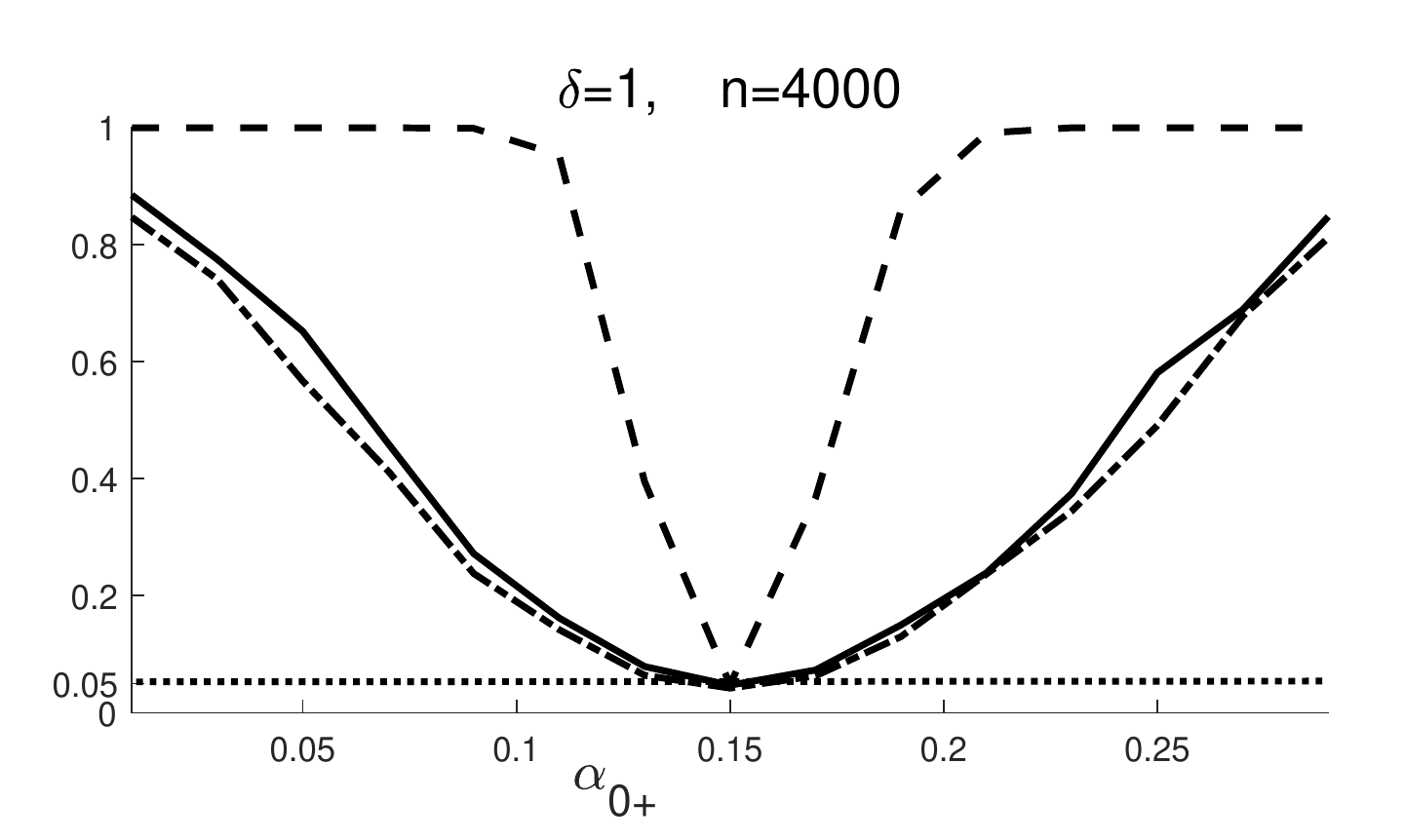}
\caption{\label{figure1}The power for the asymmetric test $\hat{S}_{1,r}$ (dotted line, - -), $\hat{S}_{2,r}^{(\tau_1)}$ (solid line, -)
 and $\hat{S}_{2,r}^{(\tau_2)}$ (solid and dotted line, -. ). Here, $r=1$, $\tau_1=0.05$, $\tau_2=0.1$, and $\eta_t\sim N(0,1)$.}
\end{center}
\end{figure}

\begin{figure}[!htbp]
\begin{center}
\includegraphics[width=10pc,height=10pc]{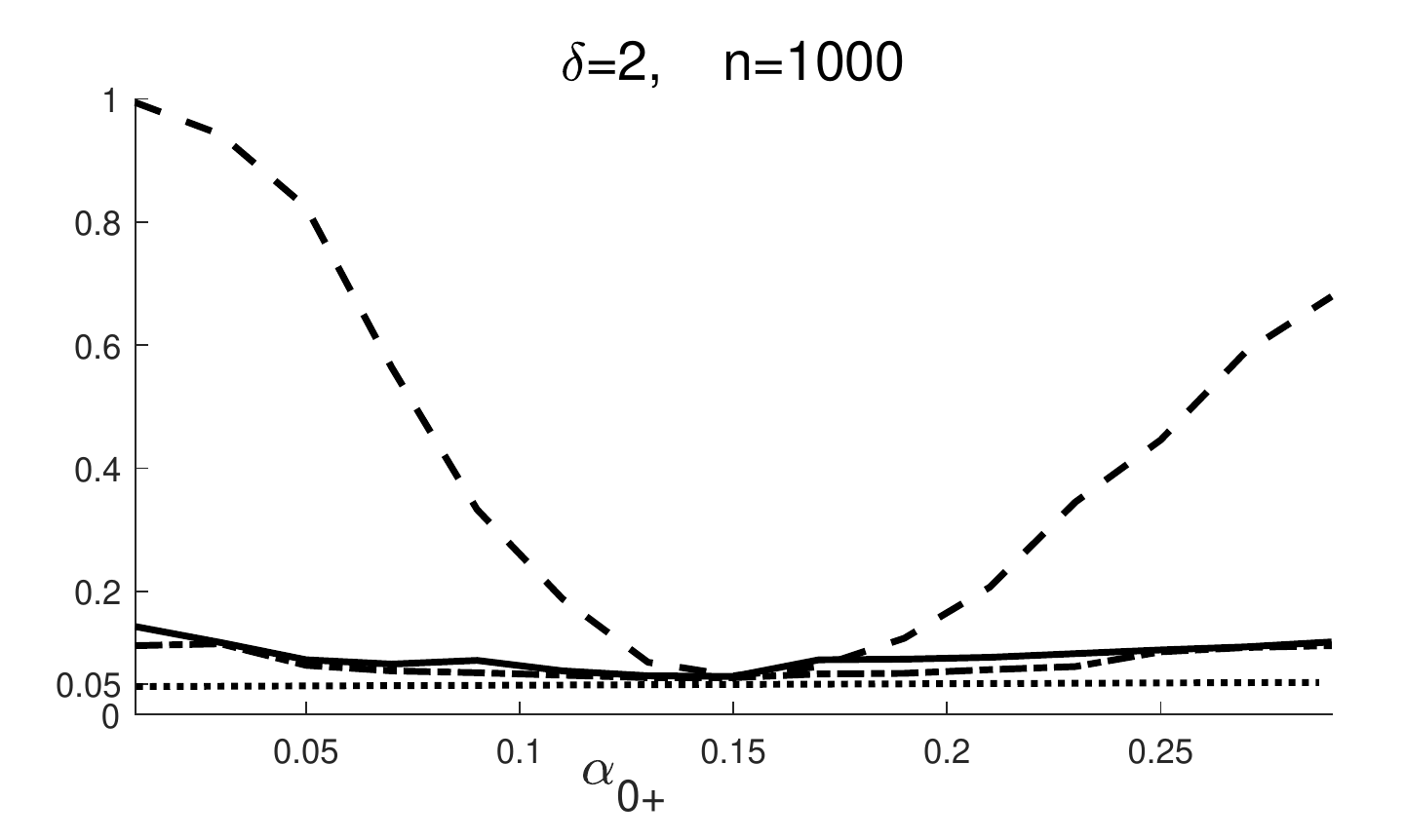}\includegraphics[width=10pc,height=10pc]{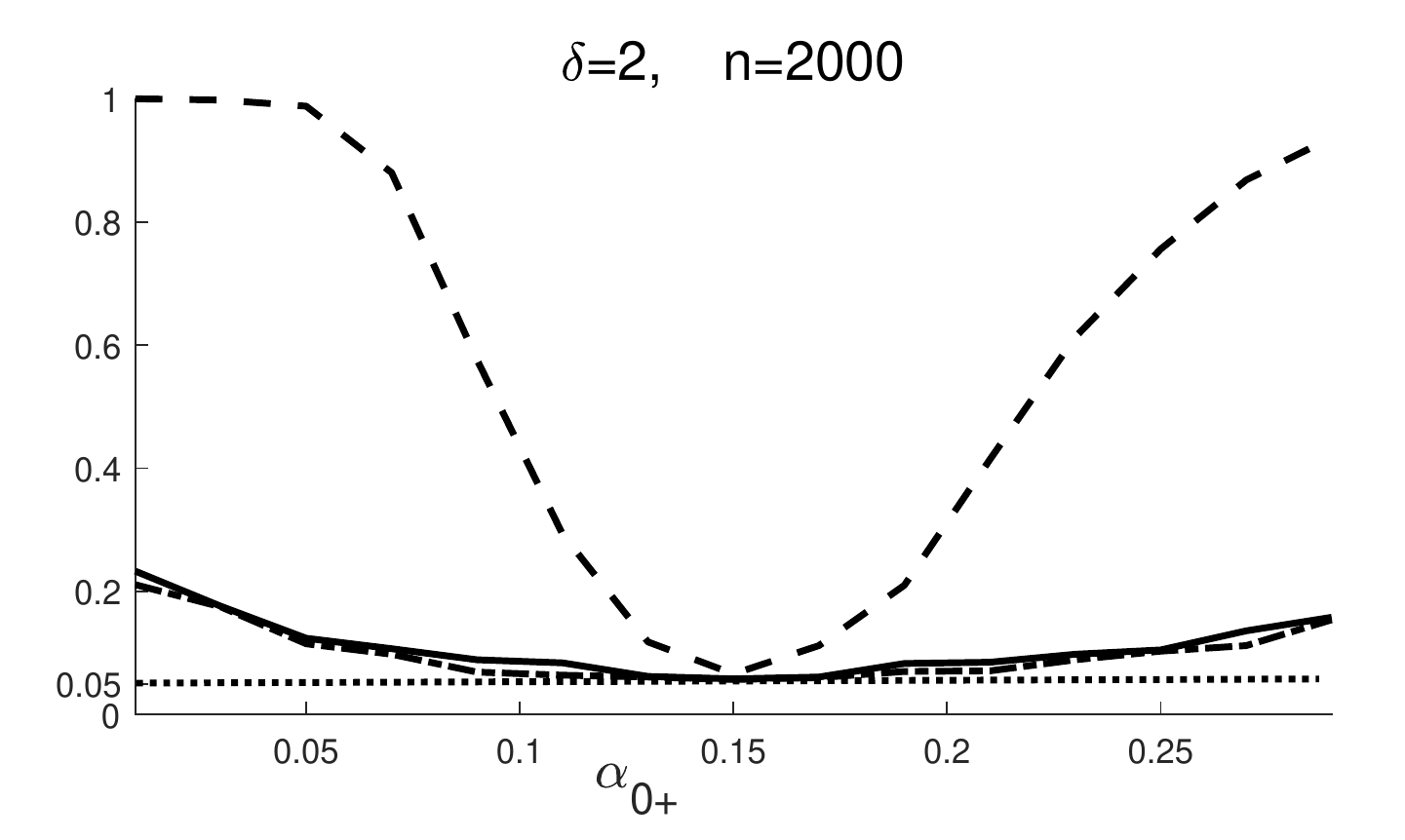}\includegraphics[width=10pc,height=10pc]{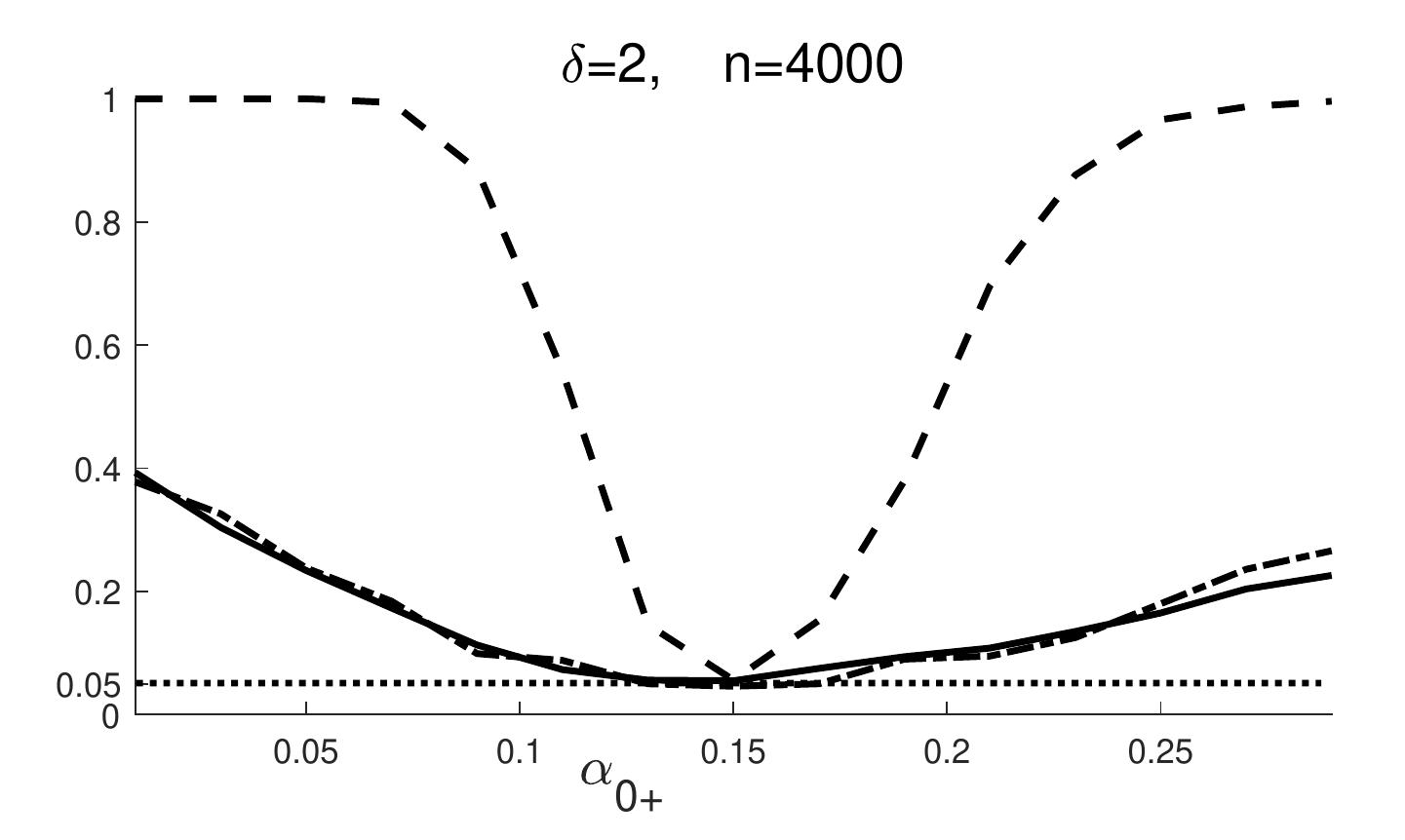}\\
\includegraphics[width=10pc,height=10pc]{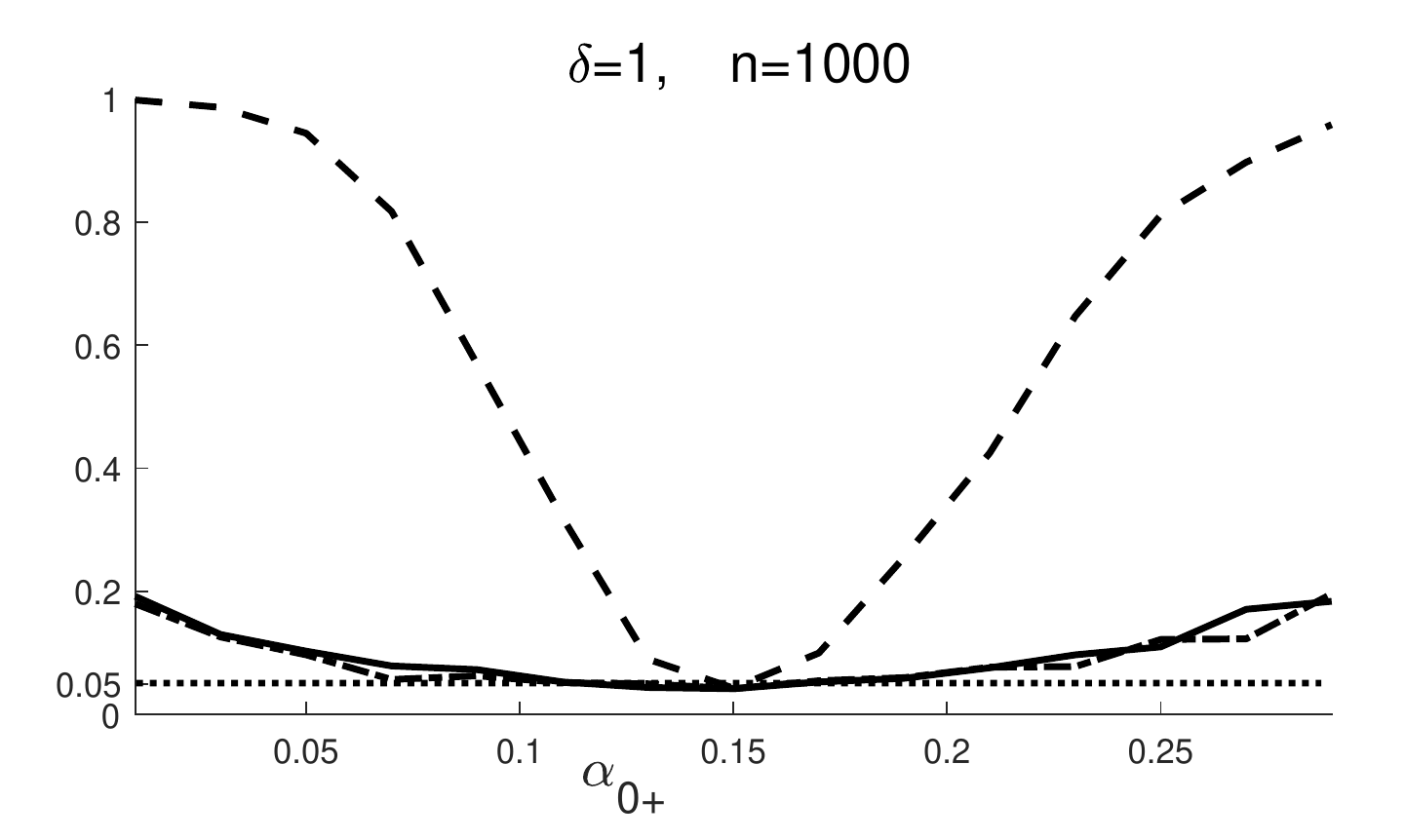}\includegraphics[width=10pc,height=10pc]{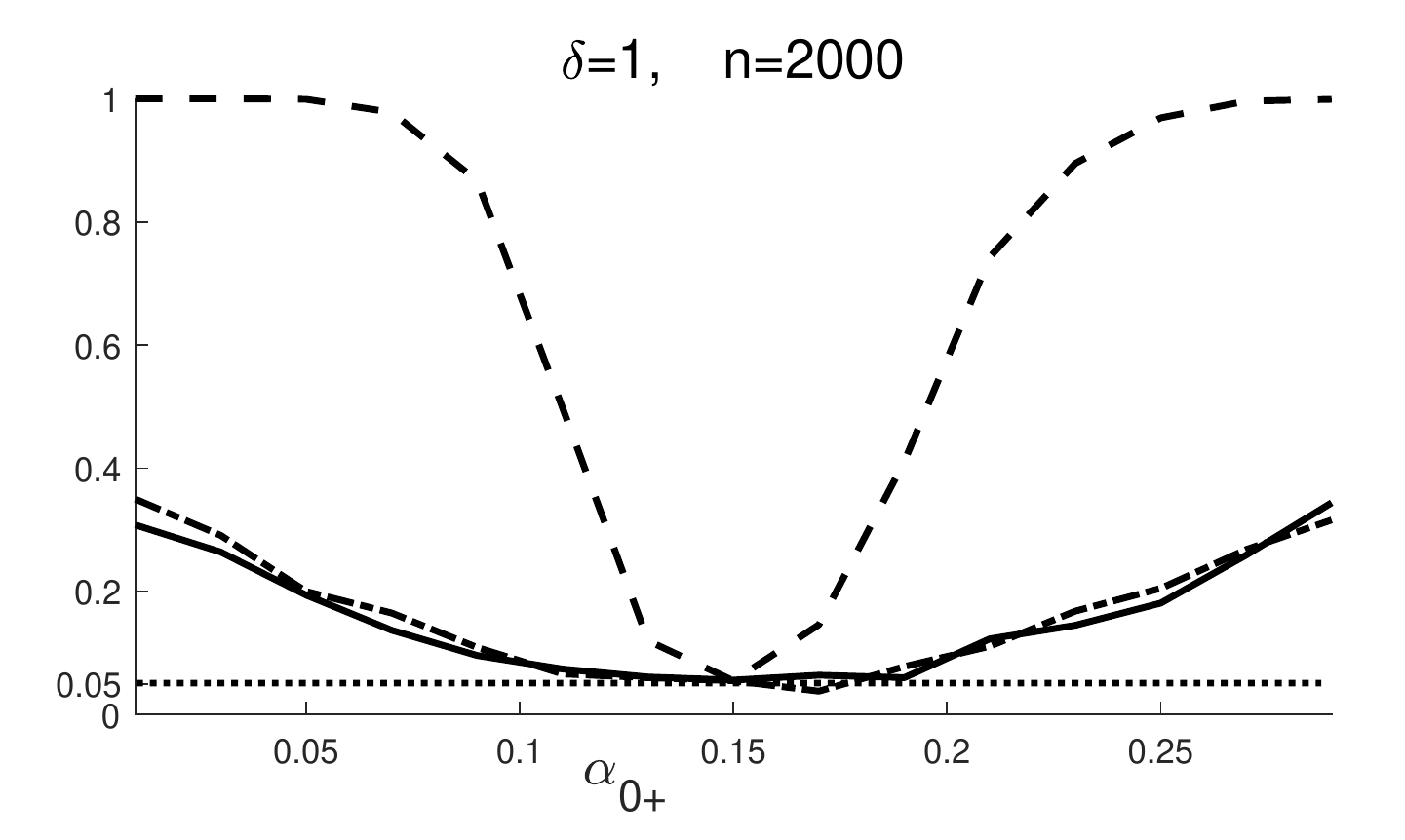}\includegraphics[width=10pc,height=10pc]{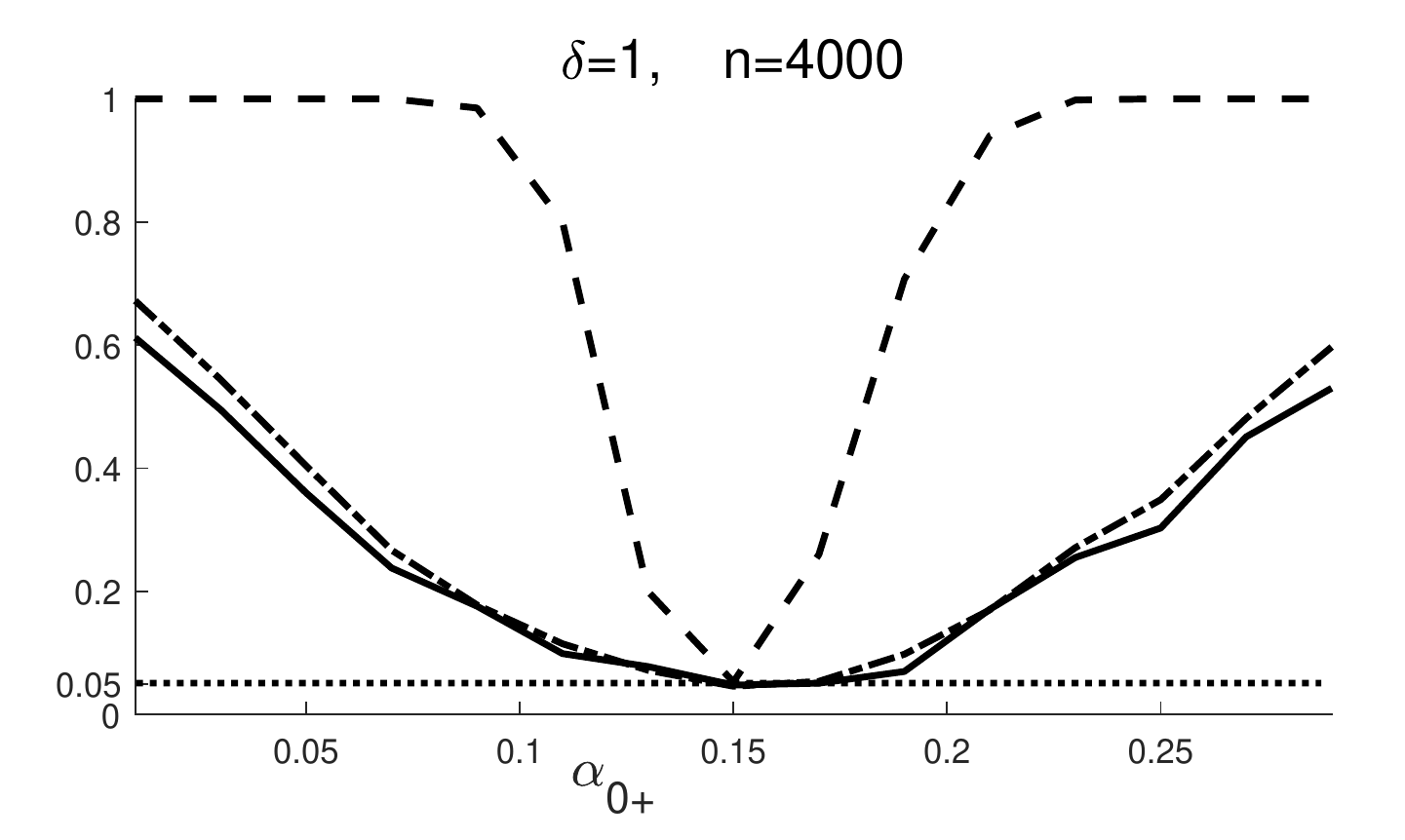}
\caption{\label{figure2}The power for the asymmetric test $\hat{S}_{1,r}$ (dotted line, - -), $\hat{S}_{2,r}^{(\tau_1)}$ (solid line, -)
 and $\hat{S}_{2,r}^{(\tau_2)}$ (solid and dotted line, -. ). Here, $r=1$, $\tau_1=0.05$, $\tau_2=0.1$, and $\eta_t\sim \mbox{st}_5$.}
\end{center}
\end{figure}

\section{Applications}
\subsection{Stationary data}
In this subsection, we re-analyze the daily log returns of two stock market indexes: the S\&P 500 index and the Dow 30 index in Zheng et al. (2018).
The data are observed on a daily basis from January 2, 2008 to June 30, 2016, with a sample size $n=2139$.
Zheng et al. (2018) studied these two datasets by using the classical GARCH($1, 1$) model, whose conditional quantile was
estimated by the hybrid quantile estimator with the Guassian QMLE as its first step estimator.
They found that the resulting method can produce better interval forecast than many existing ones.
Since their GARCH($1, 1$) model overlooks the often observed asymmetry effect in financial data, it is of interest to re-fit
these two  sequences by model (\ref{1.1}).

Based on model (\ref{1.1}) with $\delta=2$ and $1$, Table \ref{estimatedpara} gives the estimation results for both sequences.
Here, we use the GQMLE $\tilde{\theta}_{n,r}$ with $r=2$ and $1$ in the first step estimation, and
we consider the hybrid quantile estimators $\tilde{\theta}_{\tau n,r}$ with $\tau=0.05$ and $0.1$ in the second step estimation. From this table, the estimates of
$\alpha_{0+}$ are always much smaller than those of $\alpha_{0-}$ in magnitude, indicating that there is a strong asymmetric effect for both
sequences. To look for more evidence, we apply the asymmetry tests $\hat{S}_{1,r}$ and  $\hat{S}_{2,r}^{(\tau)}$ to both sequences, and
their corresponding p-values given in Table \ref{estimatedpara} confirm the asymmetric phenomenon. We also consider
the strict stationarity test $\hat{T}_{r}$ for the testing problem (\ref{4.2}) in Table \ref{estimatedpara}, and
its p-values show strong evidence that both time series are strictly stationary.

Next, we calculate the interval forecast of each sequence by the following expanding window procedure: first conduct the estimation using the data from January 2, 2008 to  December 31, 2010 and compute the conditional quantile forecast for the next trading day, i.e., the forecast of $Q_{\tau}(\epsilon_{n+1}|\mathcal{F}_n)$; then, advance the forecasting origin by one to include one more observation in the estimation subsample, and repeat the foregoing procedure until the end of the sample is reached.

\begin{table}[!htb]
\scriptsize
\caption{\label{estimatedpara}
The estimation and testing results for the S\&P 500 and  Dow 30 returns}
 \renewcommand{\arraystretch}{1.1}

\setlength{\tabcolsep}{4mm}{
\begin{tabular}{lccccccccccccccccccccc}


\hline
&\multicolumn{3}{c}{$\delta=2$} & &
\multicolumn{3}{c}{$\delta=1$}\\
\cline{2-4}\cline{6-8}

& \multicolumn{1}{c}{$r=2$} &
& \multicolumn{1}{c}{$r=1$} &
& \multicolumn{1}{c}{$r=2$} &
& \multicolumn{1}{c}{$r=1$} \\
\hline
\multicolumn{8}{c}{Panel A: S\&P 500}\\
$\omega$&4e-6 (9e-7)&& 2e-6 (4e-7)&& 7e-4 (1e-4)&&3e-4 (5e-5)\\
$\alpha_{+}$&1e-7 (0.021)&&4e-6 (0.011) &&7e-6 (0.035)&&1e-4 (0.017)\\
$\alpha_{-}$&0.261 (0.036)&&0.156 (0.018) &&0.302 (0.043) &&0.205 (0.019)\\
$\beta$&0.848 (0.025)&& 0.850 (0.018)&&0.835 (0.031) &&0.862 (0.018) \\\\

$\omega_{\tau_1}$&-1e-5 (2e-5)&& -1e-5 (2e-5)&& -1e-3 (1e-3)&&-9e-4 (1e-3)\\
$\alpha_{\tau_1+}$&-4e-7 (0.214)&&-2e-5 (0.182) &&-1e-5 (0.111) &&-4e-4 (0.113) \\
$\alpha_{\tau_1-}$&-0.812 (0.357)&&-0.872 (0.308) &&-0.517 (0.111) &&-0.476 (0.113)\\
$\beta_{\tau_1}$&-2.641 (0.004)&& -4.689 (0.003)&&-1.428 (0.172) &&-2.002 (0.230)\\\\

$\omega_{\tau_2}$&-6e-6 (7e-6)&& -5e-6 (8e-6)&& -8e-4 (9e-4)&& 0.001 (9e-4)\\
$\alpha_{\tau_2+}$&-2e-7 (0.089)&&-1e-5 (0.098) &&-9e-6 (0.086) &&-0.002 (0.082) \\
$\alpha_{\tau_2-}$&-0.431 (0.143)&&-0.456 (0.160)&&-0.388 (0.093) &&-0.175 (0.088)\\
$\beta_{\tau_2}$&-1.403 (0.002)&& -2.454 (0.002)&&-1.072 (0.130) &&-1.486 (0.154) \\\\

$\hat{T}_{r}$&1e-21&& 8e-14 && 7e-83 && 3e-51\\
$\hat{S}_{1,r}$&1e-13&& 1e-10 && 6e-15 && 7e-13\\
$\hat{S}_{2,r}^{(\tau_1)}$&0.023&& 0.006 && 5e-6 && 4e-5 \\
$\hat{S}_{2,r}^{(\tau_2)}$&0.004&& 0.006 && 2e-5 && 0.030 \\\\

\multicolumn{8}{c}{Panel B: Dow 30}\\
$\omega$&3e-6 (7e-7)&& 2e-6 (3e-7)&& 6e-4 (1e-4)&&3e-4 (5e-5)\\
$\alpha_{+}$&4e-10 (0.019)&&1e-8 (0.010)&&2e-5 (0.029)&&1e-5 (0.016)\\
$\alpha_{-}$&0.258 (0.035)&&0.160 (0.018)&&0.203 (0.037) &&0.205(0.019)\\
$\beta$&0.852 (0.021)&& 0.852 (0.018)&&0.839 (0.027)&&0.863 (0.017) \\
\\

$\omega_{\tau_1}$&-1e-5 (9e-6)&& -8e-6 (9e-6)&& -1e-3 (0.001)&&-9e-4 (1e-3)\\
$\alpha_{\tau_1+}$&-1e-9 (0.156)&&-5e-8 (0.158) &&-4e-5 (0.114)&&-2e-4 (0.122)\\
$\alpha_{\tau_1-}$&-0.784 (0.232)&&-0.862 (0.218) &&-0.501 (0.114)&&-0.474 (0.119)\\
$\beta_{\tau_1}$&-2.590 (0.002)&& -4.599 (0.002)&&-1.447 (0.172)&&-2.015 (0.230)\\
\\

$\omega_{\tau_2}$&-5e-6 (6e-6)&& -4e-6 (6e-6)&& -8e-4 (9e-4)&&-9e-4 (8e-4)\\
$\alpha_{\tau_2+}$&-7e-10 (0.095)&&-3e-8 (0.099) &&-3e-5 (0.090)&&-5e-3 (0.088)\\
$\alpha_{\tau_2-}$&-0.427 (0.154)&&-0.462 (0.166) &&-0.377 (0.098) &&-0.141 (0.095)\\
$\beta_{\tau_2}$&-1.411 (0.002)&& -2.465 (0.002)&&-1.087 (0.133) &&-1.504 (0.159)\\\\

$\hat{T}_{r}$&5e-20&& 8e-14 && 1e-83 && 2e-51\\
$\hat{S}_{1,r}$&1e-15&& 2e-10 && 5e-15 && 7e-13\\
$\hat{S}_{2,r}^{(\tau_1)}$&0.002&& 5e-4 && 7e-6 && 8e-5 \\
$\hat{S}_{2,r}^{(\tau_2)}$&0.008&& 0.007 && 1e-4 && 0.087 \\

\hline

\end{tabular}
}
 \begin{tablenotes}
   \item {\scriptsize $\dag$ Note that $\tau_1=0.05$ and $\tau_2=0.1$.\\
        $\ddag$ The standard deviations of all estimators are given in parentheses, and the p-values of all tests are given.}
 \end{tablenotes}

\end{table}

Moreover, we evaluate the forecasting performance of the aforementioned interval forecasts by using the following two measures:

(i) the minimum of the p-values of the two VaR backtests, the likelihood ratio test for correct conditional converge (CC) in Christoffersen (1998) and the dynamic quantile (DQ) test\footnote{As in Zheng et al. (2008), the regression matrix contains four lagged hits and the contemporaneous VaR estimate
 for DQ test.} in Engle and Manganelli (2004);

(ii) the empirical coverage error is defined as the proportion of observations that exceed the corresponding VaR forecast minus the corresponding nominal level $\tau$.

 The reason for selecting the smaller of the two p-values is that the CC and DQ tests have different null hypotheses and hence are complementary to each other. Note that a larger p-value of either CC or DQ test gives a stronger evidence of good interval forecasts.

 \begin{sidewaystable}
\caption{\label{forecastresult}Minimum p-values of two VaR backtests and empirical coverage errors for the S\&P 500 and Dow 30 returns at the lower (L) (or upper (U)) 0.01th, 0.025th, and 0.05th conditional quantiles}

\setlength{\tabcolsep}{2.5mm}{
\begin{tabular}{llccccccccccccccccccccc ccccc ccccc}

\hline

&
& \multicolumn{7}{c}{Minimum p-value of VaR backtests} &
& \multicolumn{7}{c}{Empirical coverage error}\\

\cline{3-9}\cline{11-17}

&&\multicolumn{4}{c}{$\delta=2$}&
& \multicolumn{2}{c}{$\delta=1$}&
& \multicolumn{4}{c}{$\delta=2$}&
& \multicolumn{2}{c}{$\delta=1$}\\

\cline{3-6}\cline{8-9}\cline{11-14}\cline{16-17}

& &\multicolumn{3}{c}{$r=2$} & & & &  &&\multicolumn{3}{c}{$r=2$}& && & \\

\cline{3-5} \cline{11-13}

& \multicolumn{1}{c}{$\tau$} &\multicolumn{1}{c}{$\alpha_{0+}=\alpha_{0-}$}&
& $\alpha_{0+}\not=\alpha_{0-}$ & $r=1$ & &$r=2$ & $r=1$ &&\multicolumn{1}{c}{$\alpha_{0+}=\alpha_{0-}$}&&$\alpha_{0+}\not=\alpha_{0-}$& $r=1$&& $r=2$& $r=1$\\

\hline\\

S\&P 500&L1.0&0.000&&0.0000&0.0000&&0.0000&0.0000&&-0.0002&&-0.0069&-0.0069&&-0.0088&-0.0076\\
&L2.5&0.001&&0.0000&0.0000&&0.0000&0.0000&&-0.0048& &-0.0226&-0.0195&&-0.0183&-0.0183\\
&L5.0&0.017&&0.0000&0.0000&&0.0000&0.0000&&-0.0090 &&-0.0427&-0.0378&&-0.0323&-0.0255\\
&U5.0&0.245&&\textbf{0.6996}&0.6304&&0.4846&0.2401&&0.0054&&0.0041&0.0023&&0.0047&\textbf{0.0011}\\
&U2.5&0.356&&0.7142&\textbf{0.7616}&&0.1476&0.2807&&0.0030 &&0.0030&\textbf{0.0011}&&0.0060&0.0048\\
&U1.0&0.275&&\textbf{0.8504}&0.2956&&0.8213&0.6206&&0.0008 &&\textbf{0.0002}&-0.0028&&0.0008&0.0020\\\\

Dow 30&L1.0&0.063&&0.0000&0.0000&&0.000& 0.0000&&-0.0014 &&-0.0076&-0.0027&&-0.0088&-0.0076\\
&L2.5&0.000&&0.0000&0.0000&&0.0000&0.0000&&-0.0054 &&-0.0249&-0.0201&&-0.0213&-0.213\\
&L5.0&0.000&&0.0000&0.0000&&0.000&0.0000&&-0.0072 &&-0.0433&-0.0420&&-0.0329&-0.0286\\
&U5.0&0.273&&0.1678&0.2304&&0.1842&\textbf{0.2798}&&0.0084&&0.0072&0.0035&&0.0060&\textbf{0.0023}\\
&U2.5&0.568&&0.3493&0.3723&&\textbf{0.6350}&0.3723&&0.0011&&0.0024&0.0005&&0.0011&\textbf{0.0001}\\
&U1.0&0.418&&\textbf{0.8256}&\textbf{0.8256}&&0.1296&0.0002&&-0.0028&&\textbf{-0.0004}&\textbf{-0.0004}&&0.0045&0.0020\\

\hline
\end{tabular}
}
 \begin{tablenotes}
     \item {\scriptsize $\dag$ Among the models with p-values $>5\%$, the largest p-value and the smallest empirical coverage error (in absolute value) are in boldface.}
 \end{tablenotes}
\end{sidewaystable}

Based on model (\ref{1.1}) with $\delta=2$ and 1, Table \ref{forecastresult} reports the results of two measures at the lower (L) (or upper(U)) 0.01th, 0.025th and 0.05th conditional quantiles.
Here, the GQMLE $\tilde{\theta}_{n,r}$ with $r=2$ and $1$ is used in the first step estimation.
As a comparison, the results for the benchmark method (i.e., $\delta=2, r=2$ and $\alpha_{0+}=\alpha_{0-}$) in
Zheng et al. (2018) are also included in Table \ref{forecastresult}.
It can be seen that all methods have a poor performance for the lower conditional quantiles, while our proposed methods, based on
the asymmetric model (\ref{1.1}) together with the hybrid quantile estimation,
have a significantly better interval forecasting performance for the upper conditional quantiles than the benchmark method in Zheng et al. (2018).
The poor performance of the lower conditional quantiles from our method may be because our GQMLE $\tilde{\theta}_{n,r}$ does not account for the asymmetry of $\eta_t$. We
may expect to improve our forecasting performance particularly for the lower conditional quantiles by using a skewed distribution of $\eta_t$ to form our first estimation, and we leave this desired direction for future study.
In terms of the minimum of the p-values of the two VaR backtests, our proposed methods with $\delta=2$ are better than those with $\delta=1$
in four out of six cases\footnote{Only consider the cases that the minimum of the p-values of two backtests is larger than 5\%},  while the choice of $r$ seems irrelevant to the forecasting performance. In terms of the empirical coverage error,
our proposed methods with $\delta=2$ (or $r=1$) are better than those with $\delta=1$ (or $r=2$) in general.
Overall, our method with $\delta=2$, $r=2$ and $\alpha_{0+}\not=\alpha_{0-}$ has the best interval forecasting performance for
both data.

\subsection{Non-stationary data}

 In this subsection, we re-visit three daily stock return data sequences of Community Bankers Trust (BTC), China MediaExpress (CCME) and Monarch Community Bancorp (MCBF)
  in Francq and Zako\"{i}an (2012, 2013a). These three sequences are shown to be non-stationary in Francq and Zako\"{i}an (2012), while their conditional quantile estimators have not been investigated. Motivated by this, we study their conditional quantiles by our hybrid quantile estimation method. To compute our hybrid quantile estimator, we
  choose the GQMLE $\tilde{\theta}_{n,r}$ with $r=1$ in the first estimation step. Here, we do not consider the GQMLE $\tilde{\theta}_{n,r}$ with $r=2$, since Li et al. (2018) demonstrated
  the innovations of the fitted GARCH($1, 1$) model for each sequence only have a finite second moment but not an infinite fourth moment. In the
  second step of quantile estimation, we consider the hybrid quantile estimators $\hat{\theta}_{\tau n,1}$ at levels $\tau=0.05$ and $0.1$.
  Table \ref{estimatedpara2} reports the results of $\tilde{\theta}_{n,1}$ and $\hat{\theta}_{\tau n,1}$ for each sequence, together with the results of
  $\hat{T}_{1}$ for the testing problem (\ref{4.2}). From the results of $\hat{T}_{1}$, we can reach the same conclusion as in Francq and Zako\"{i}an (2012) that
  all three data are non-stationary, and hence the estimates for the drift term $\omega$ or $\omega_{\tau}$ may not be consistent.
  Meanwhile, Table \ref{estimatedpara2} reports the results of $\hat{S}_{1,1}$, $\hat{S}_{2, 1}^{(0.05)}$ and $\hat{S}_{2,1}^{(0.1)}$
  for the testing problem (\ref{4.6}). It is interesting to observe that the global asymmetry test $\hat{S}_{1,1}$ as the one in Francq and Zako\"{i}an (2013a) indicates that all three datasets do not have the asymmetric effect, while
  the local asymmetry tests  $\hat{S}_{2, 1}^{(0.05)}$ and $\hat{S}_{2,1}^{(0.1)}$ detect some strong asymmetric effects in model (\ref{1.1}) with $\delta=2$ or $1$ for the CCME and MCBF data. Although none of the considered tests can find the asymmetric evidence for the BTC data, we think the examined BTC data still have the asymmetric effect, since our
  forecasting comparison below indicates that the asymmetric PGARCH model can perform better than its symmetric counterpart.

  \begin{table}[!ht]\scriptsize
\caption{\label{estimatedpara2}
The estimation and testing results for the BTC, CCME and MCBF returns}
\addtolength{\tabcolsep}{-3pt}
 \renewcommand{\arraystretch}{1.1}
\setlength{\tabcolsep}{1.3mm}{
\begin{tabular}{lccccccccccc}


\hline

&\multicolumn{3}{c}{Panel A: BTC} & &
\multicolumn{3}{c}{Panel B: CCME} & &
\multicolumn{3}{c}{Panel C: MCBF}\\

& \multicolumn{1}{c}{$\delta=2$} &
& \multicolumn{1}{c}{$\delta=1$} &
& \multicolumn{1}{c}{$\delta=2$} &
& \multicolumn{1}{c}{$\delta=1$} &
& \multicolumn{1}{c}{$\delta=2$} &
& \multicolumn{1}{c}{$\delta=1$} \\

\cline{2-4}\cline{6-8}\cline{10-12}

$\omega$&8e-7 (7e-8)&& 1e-4 (1e-4) &&2e-8 (2e-8) &&1e-4 (2e-5)&&8e-6 (4e-6) &&8e-4 (3e-4)\\
$\alpha_{+}$&0.089 (0.035) &&0.130 (0.040) &&0.107 (0.047) &&0.148 (0.048)&&0.033 (0.016) &&0.078 (0.003)\\
$\alpha_{-}$&0.119 (0.038) &&0.172 (0.041) &&0.125 (0.063) &&0.161 (0.056)&&0.029 (0.014)  &&0.078 (0.028)\\
$\beta$&0.840 (0.031) &&0.854 (0.027) &&0.838 (0.043) &&0.860 (0.033)&&0.931 (0.019) &&0.902 (0.024)\\
\\
$\omega_{\tau_1}$&-5e-7 (1e-6) &&-9e-4 (4e-4) &&-1e-9 (6e-7) &&-1e-7 (3e-4)&&-2e-7 (1e-4) &&-9e-4 (0.005)\\
$\alpha_{\tau_1+}$&-0.448 (0.229) &&-0.320 (0.211) &&-0.639 (0.421) &&-0.498 (0.295)&&-1.515 (0.221) &&-0.479 (0.346)\\
$\alpha_{\tau_1-}$&-0.661 (0.215) &&-0.423 (0.182) &&-1.879 (0.504) &&-0.846 (0.298)&& -1e-4 (0.086) &&-0.009 (0.144)\\
$\beta_{\tau_1}$&-4.660 (1e-4) &&-2.097 (0.060) &&-3.190 (1e-4) &&-1.772 (0.032)&& -4.625 (0.009) &&-2.037 (0.263)\\
\\
$\omega_{\tau_2}$&-8e-8 (7e-7) &&-2e-7 (3e-4) &&-4e-13 (2e-7) &&-3e-8 (2e-4)&&-1e-8 (9e-5) &&-1e-4 (0.003)\\
$\alpha_{\tau_2+}$&-0.348 (0.211) &&-0.153 (0.140) &&-0.364 (0.268) &&-0.404 (0.191)&& -0.596 (0.229) &&-0.314 (0.173)\\
$\alpha_{\tau_2-}$&-0.198 (0.178) &&-0.113 (0.106) &&-0.741 (0.267) &&-0.792 (0.182)&&-2e-5 (0.088) &&0.005 (0.095)\\
$\beta_{\tau_2}$&-2.232 (1e-4) &&-1.522 (0.004) &&-1.450 (1e-4) &&-0.948 (0.021)&& -2.438 (0.008) &&-1.534 (0.151)\\\\

$\hat{T}_{r}$&0.397 &&0.966 &&0.145 &&0.577 &&0.894 &&0.143\\

$\hat{S}_{1,r}$&0.222 &&0.208 &&0.409 &&0.424 &&0.429 &&0.499\\

$\hat{S}_{2,r}^{(\tau_1)}$&0.257 &&0.359 &&0.032 &&0.210 &&1e-4 &&0.063\\

$\hat{S}_{2,r}^{(\tau_2)}$&0.298 &&0.411 &&0.164 &&0.076 &&0.008 &&0.035\\

\hline
\end{tabular}
}
 \begin{tablenotes}
   \item {\scriptsize $\dag$ Note that $r=1$, $\tau_1=0.05$ and $\tau_2=0.1$.\\
        $\ddag$ The standard deviations of all estimators are given in parentheses, and the p-values of all tests are given.}
 \end{tablenotes}
\end{table}

\begin{table}[!ht]
\caption{\label{forecastresult2}Minimum p-values of two VaR backtests and empirical coverage errors for the BTC, CCME and MCBF returns at the lower (L) (or upper (U)) 0.01th, 0.025th, and 0.05th conditional quantiles}
\small\addtolength{\tabcolsep}{-2pt}
\begin{tabular}{lcccccccccccc}
\hline
&
& \multicolumn{5}{c}{Minimum p-value of VaR backtests} &
& \multicolumn{5}{c}{Empirical coverage error}\\

\cline{3-7}\cline{9-13}

& &\multicolumn{3}{c}{$\delta=2$}&&  &&\multicolumn{3}{c}{$\delta=2$}&& \\

\cline{3-5}\cline{9-11}

& \multicolumn{1}{c}{$\tau$}&\multicolumn{1}{c}{$\alpha_{0+}=\alpha_{0-}$}&& $\alpha_{0+}\not=\alpha_{0-}$ && $\delta=1$ &&\multicolumn{1}{c}{$\alpha_{0+}=\alpha_{0-}$}&&$\alpha_{0+}\not=\alpha_{0-}$&& $\delta=1$\\
\hline\\

BTC&L1.0&0.0025 &&0.3999 &&\textbf{0.9329}&& -0.0100 &&-0.0056 &&\textbf{-0.0012} \\
&L2.5&0.0000 &&0.0025&&\textbf{0.3335}&& -0.0250 &&-0.0206&&\textbf{-0.0096}\\
&L5.0&0.0000 &&0.0000&&0.0031&& -0.0302 &&-0.0478 &&-0.0302\\
&U5.0& 0.0299&&\textbf{0.1265}&&0.0182&& 0.0500 &&\textbf{0.0170}&&0.0192\\
&U2.5& 0.0003 &&\textbf{0.6372}&&0.0877&& 0.0228 &&0.0052&&\textbf{0.0008}\\
&U1.0& 0.1296 &&\textbf{0.8130}&&0.2569&& 0.0038 &&\textbf{-0.0010}&&-0.0032\\\\

CCME&L1.0& 0.0301&&\textbf{0.9574}&&\textbf{0.9574}&&-0.0100 &&\textbf{-0.0015}&&\textbf{-0.0015} \\
&L2.5& 0.0006 &&0.0001&&0.0000&& -0.0250 &&-0.0165&&-0.0165\\
&L5.0&0.0000 &&0.0000&&0.0000&&-0.0500 &&-0.0415&&-0.0415\\
&U5.0&0.0002 &&0.0000&&0.0080&& 0.0457 &&0.0457&&0.0372\\
&U2.5&0.0433 &&0.0006&&0.0006&&0.0207 &&0.0250&&0.0250\\
&U1.0&\textbf{0.6077} &&\textbf{0.6077}&&0.0301&&\textbf{0.0057} &&\textbf{0.0057}&&0.0100\\\\

MCBF&L1.0& 0.0031&&0.0031&&0.0031&&-0.0100 &&-0.0100&&-0.0100 \\
&L2.5&0.0038 &&0.2131&&\textbf{0.4400}&&-0.0204 &&-0.0112&&\textbf{0.0050}\\
&L5.0&0.0023 &&0.1220&&\textbf{0.1682}&&-0.0316 &&-0.0177&&\textbf{-0.0131}\\
&U5.0&0.0067&&0.0023&&0.0001&&0.0200 &&0.0316&&-0.0030\\
&U2.5&0.0005 &&\textbf{0.7622}&&0.0001&&0.0227 &&\textbf{0.0020}&&-0.0188\\
&U1.0&0.0031 &&0.0000&&\textbf{0.7747}&&0.0100 &&\textbf{0.0008}&&0.0031\\
\hline
\end{tabular}
 \begin{tablenotes}
     \item {\scriptsize $\dag$ Among the models with p-values $>5\%$, the largest p-value and the smallest empirical coverage error (in absolute value) are in boldface.}
 \end{tablenotes}
\end{table}

  Next, we compute the interval forecasts for each sequence by using the same procedure as in Subsection 6.1, except that
  the first interval forecast is calculated based on the first half of sample. Again,
  we follow the measurements as in Subsection 6.1 to evaluate the interval forecasting performance of our methods, based on
  model (\ref{1.1}) with the hybrid quantile estimators. Table \ref{forecastresult2} reports the corresponding results
  for all three datasets. As a comparison, the forecasting performance of
  the benchmark GARCH($1, 1$) model (i.e., $\delta=2$ and $\alpha_{0+}=\alpha_{0-}$) estimated by the Laplacian QMLE $\tilde{\theta}_{n,1}$ is also given in Table \ref{forecastresult2}. It can be seen that, in terms of minimum p-values of two VaR backtests, model (\ref{1.1}) with $\delta=1$ (or $\delta=2$ and $\alpha_{0+}\not=\alpha_{0-}$)
  can provide us with a good interval forecast in 6 cases,  while the benchmark GARCH($1, 1$) model can only do this in one case.
  Similar conclusions can be obtained in terms of empirical coverage error.
  Particularly, our forecasting results indicate that the BTC data have the asymmetric effect, which, however, has not been detected by our considered tests in Table
  \ref{estimatedpara2}. Note that there are 7 cases (most of them are for the CCME data) in which none of the
  methods can deliver a satisfactory interval forecast, and these cases may require some new methods for their interval forecast.

\section{Conclusion}

In this paper, the hybrid quantile estimators are proposed for the asymmetric PGARCH models via the transformation
$T(x)=|x|^{\delta}\mbox{sgn}(x)$. Asymptotic normality for the quantile estimators is established under both stationarity and
non-stationarity. As a result, tests for strict stationarity and asymmetry are obtained. It is hoped these results will add to the tool kits of time series analysis.

\section*{Acknowledgement}
The authors greatly appreciate the very helpful comments and suggestions of two anonymous reviewers and the editor.
The first author's work is supported by the Fundamental Research Funds for the Center University (12619624).
The second author's work is supported by RGC of Hong Kong (Nos. 17306818 and 17305619), NSFC (Nos. 11571348, 11690014, 11731015 and 71532013), Seed Fund
for Basic Research (No. 201811159049), and the Fundamental Research Funds for the Central University (19JNYH08).
The third author's work is supported by RGC of Hong Kong (Nos. 17304417 and 17304617).
The fourth author's work is supported by RGC of Hong Kong (No. 17304417).

\appendix

\section*{Appendix: Proofs}


To facilitate our proofs, we first introduce some notations.
Let $\Theta_0=\{\theta\in \Theta: \beta<e^{\gamma_0}\}$ and $\Theta_p=\{\theta\in \mathcal{R}_{+}^4: \beta<\|1/a_0(\eta_t)\|_p^{-1}\}$.
Define four  $[0,\infty]$-valued processes
 \begin{flalign*}
 &v_t(\vartheta)=\sum_{j=1}^{\infty}\frac{\{\alpha_+(\eta_{t-j}^+)^{\delta}+\alpha_-(-\eta_{t-j}^-)^{\delta}\}}{a_0(\eta_{t-j})}\prod_{k=1}^{j-1}\frac{\beta}{a_0(\eta_{t-k})},\\
 &d_t^{\alpha_+}(\vartheta)=\sum_{j=1}^{\infty}\frac{(\eta_{t-j}^+)^{\delta}}{a_0(\eta_{t-j})}\prod_{k=1}^{j-1}\frac{\beta}{a_0(\eta_{t-k})},
 \ \ \ d_t^{\alpha_-}(\vartheta)=\sum_{j=1}^{\infty}\frac{(-\eta_{t-j}^-)^{\delta}}{a_0(\eta_{t-j})}\prod_{k=1}^{j-1}\frac{\beta}{a_0(\eta_{t-k})},\\
 &d_t^{\beta}(\vartheta)=\sum_{j=2}^{\infty}\frac{(j-1)\{\alpha_+(\eta_{t-j}^+)^{\delta}+\alpha_-(-\eta_{t-j}^-)^{\delta}\}}{\beta a_0(\eta_{t-j})}\prod_{k=1}^{j-1}\frac{\beta}{a_0(\eta_{t-k})}
 \end{flalign*}
with the convention $\prod_{k=1}^{j-1}=1$ when $j\leq 1$. As shown in Francq and Zako\"{i}an (2013a),
$v_t(\vartheta)$, $1/v_{t}(\vartheta)$, $d_t^{\alpha_+}(\vartheta)$, $d_t^{\alpha_-}(\vartheta)$ and $d_t^{\beta}(\vartheta)$ have moments of any order.

Second, we give six technical lemmas. Lemmas \ref{lema1}-\ref{lema2} from Francq and Zako\"{i}an (2013a)
show that, after being normalized by $h_t$, the nonstationary process $\sigma_{t}^{\delta}(\theta)$ and its first derivatives can be well
approximated by some stationary processes. Lemma \ref{lema3} gives the asymptotic properties of
the GQMLE $\tilde{\theta}_{n,r}$, and its proof is similar to that of Theorem 3.1 in Francq and Zako\"{i}an (2013a).
Lemma \ref{lema4} proves the consistency of $\tilde{\vartheta}_{n,r}$ for $\gamma_0\geq 0$.
Lemmas \ref{lema5}-\ref{lema6} are used to for the proof of Theorem \ref{thm3.1}.

\begin{lem}\label{lema1}
Suppose that Assumption \ref{asm3.1}(ii) holds.

(i) When $\gamma_0>0$, for any $\theta\in \Theta_0$, the process $v_t(\vartheta)$ is stationary and ergodic. Moreover, for any compact set $\Theta_0^*\subset \Theta_0$,
\begin{flalign*}
\sup_{\theta\in \Theta_0^*}\left|\frac{\sigma_t^{\delta}(\theta)}{h_t}-v_t(\vartheta)\right|\rightarrow 0 \mbox{ a.s. as $t\rightarrow \infty$},
\end{flalign*}
and
\begin{flalign*}
\sup_{\theta\in \Theta_0^*}\left|\frac{h_t}{\sigma_t^{\delta}(\theta)}-\frac{1}{v_t(\vartheta)}\right|\rightarrow 0 \mbox{ a.s. as $t\rightarrow \infty$}.
\end{flalign*}
Finally, for any  $\theta\not\in \Theta_0$, it holds that $\sigma_t^{\delta}(\theta)/h_t\to\infty$ as $t\to\infty$.

(ii) When $\gamma_0=0$, for any $\theta\in \Theta_p$ with $p\geq 1$, the process $v_t(\vartheta)$ is stationary and ergodic. Moreover, for any compact set $\Theta_p^*\subset \Theta_p$,
\begin{flalign*}
\sup_{\theta\in \Theta_p^*}\left|\frac{\sigma_t^{\delta}(\theta)}{h_t}-v_t(\vartheta)\right|\rightarrow 0 \mbox{ in $L^p$ as $t\rightarrow \infty$},
\end{flalign*}
and
\begin{flalign*}
\sup_{\theta\in \Theta_p^*}\left|\frac{h_t}{\sigma_t^{\delta}(\theta)}-\frac{1}{v_t(\vartheta)}\right|\rightarrow 0 \mbox{ in $L^p$ as $t\rightarrow \infty$}.
\end{flalign*}
 \end{lem}

\begin{lem}\label{lema2}
Suppose that Assumption \ref{asm3.1}(ii) holds.

(i) When $\gamma_0>0$, for any $\theta\in \Theta_0$, the processes $d_t^{\alpha_+}(\vartheta)$, $d_t^{\alpha_-}(\vartheta)$, and $d_t^{\beta}(\vartheta)$
are stationary and ergodic. Moreover, for any compact set $\Theta_0^*\subset \Theta_0$,
\begin{flalign*}
\sup_{\theta\in \Theta_0^*}\left\|\frac{1}{h_t}\frac{\partial \sigma_{t}^{\delta}(\theta)}{\partial\vartheta}-d_t(\vartheta)\right\|\rightarrow 0 \mbox{ a.s. as $t\rightarrow \infty$},
\end{flalign*}
where
\begin{flalign}\label{a.1}
d_t(\vartheta)=(d_t^{\alpha_+}(\vartheta),d_t^{\alpha_-}(\vartheta),d_t^{\beta}(\vartheta))'.
\end{flalign}


(ii) When $\gamma_0=0$, for any $\theta\in \Theta_p$ with $p\geq 1$, the processes $d_t^{\alpha_+}(\vartheta)$, $d_t^{\alpha_-}(\vartheta)$, and $d_t^{\beta}(\vartheta)$  are stationary and ergodic. Moreover, for any compact set $\Theta_p^*\subset \Theta_p$,
\begin{flalign*}
\sup_{\theta\in \Theta_p^*}\left\|\frac{1}{h_t}\frac{\partial \sigma_{t}^{\delta}(\theta)}{\partial\vartheta}-d_t(\vartheta)\right\|\rightarrow 0 \mbox{ in}\  L^p\  \mbox{as}\
  t\rightarrow \infty.
\end{flalign*}
\end{lem}

\begin{lem}\label{lema3}
Suppose that Assumption \ref{asm3.1} holds and $E|\eta_t|^{2r} <\infty$.

(i) When $\gamma_0<0$,  and $\beta<1$ for all $\theta\in \Theta$, then $\tilde{\theta}_{ n,r}\rightarrow \theta_{0}$ a.s. as $n\rightarrow \infty$,
and
\begin{flalign}\label{a.2}
\sqrt{n}(\tilde{\theta}_{n,r}-\theta_{0})
&=-\frac{\delta J^{-1}}{r\sqrt{n}}\sum_{t=1}^{n}[1-|\eta_t|^r]\frac{1}{h_{t}}\frac{\partial h_{t}(\theta_0)}{\partial\theta}+o_{p}(1)\nonumber\\
&\rightarrow_d N(0,\kappa_{2r}\delta^2 J^{-1}) \mbox{ as }n\to\infty.
\end{flalign}

(ii) When $\gamma_0>0$, and $P(\eta_t=0)=0$, then
$\tilde{\vartheta}_{n,r}\rightarrow \vartheta_{0}$ a.s. as $n\rightarrow \infty$, and
\begin{flalign}\label{a.3}
\sqrt{n}(\tilde{\vartheta}_{ n,r}-\vartheta_{0})
&=-\frac{\delta J_{\vartheta}^{-1}}{r\sqrt{n}}\sum_{t=1}^{n}[1-|\eta_t|^r]\frac{1}{h_{t}}\frac{\partial \sigma_{t}^{\delta}(\theta_0)}{\partial\vartheta}+o_{p}(1)\nonumber\\
&\rightarrow _d N(0,\kappa_{2r}\delta^2 J_{\vartheta}^{-1}) \mbox{ as }n\to\infty.
\end{flalign}

(iii) When $\gamma_0=0$,  $P(\eta_t=0)=0$, and $\beta<\|1/a_0(\eta_t)\|_p^{-1}$ for any $\theta\in \Theta$ and  some $p>1$, then
$\tilde{\vartheta}_{ n,r}\rightarrow \vartheta_{0}$ in probability as $n\rightarrow \infty$, and (\ref{a.3}) holds provided that
Assumption \ref{asm3.3} is satisfied.
\end{lem}

\begin{lem}\label{lema4}
Suppose that Assumptions \ref{asm3.1}-\ref{asm3.2} hold and $E|\eta_t|^{2r} <\infty$.


(i) When $\gamma_0>0$, and $P(\eta_t=0)=0$, then
     $\hat{\vartheta}_{\tau n,r}\rightarrow \vartheta_{\tau0}$ in probability as $n\rightarrow \infty$.

(ii)  When $\gamma_0=0$,  $P(\eta_t=0)=0$, and $\beta<\|1/a_0(\eta_t)\|_p^{-1}$ for any $\theta\in \Theta$ and  some $p>1$, then
 $\hat{\vartheta}_{\tau n,r}\rightarrow \vartheta_{\tau0}$ in probability as $n\rightarrow \infty$.
\end{lem}

\begin{proof}
We only show the proof of (i), and the proof of (ii) is similar.

First, by (\ref{2.3}), it is straightforward that $(\hat{\omega}_{\tau n,r},\hat{\vartheta}_{\tau n,r}')'=\argmin_{\theta_\tau\in \Theta_{\tau}}Q_n(\theta_{\tau})$, where $Q_n(\theta_{\tau})=\f{1}{n}\sum_{t=1}^n[l_{t,\rho}(\theta_{\tau})-l_{t,\rho}^{\dag}]$ with
$l_{t,\rho}^{\dag}=\rho_{\tau}\big(\f{y_t}{\sigma_t^{\delta}(\tilde{\theta}_{n,r})}
-b_{\tau}\big)$. By using the identity
\begin{flalign*}
\rho_{\tau}(x-y)-\rho_{\tau}(x)=-y\psi_{\tau}(x)+\int_{0}^{y}\left[\mbox{I}(x\leq s)-\mbox{I}(x\leq 0)\right]ds
\end{flalign*}
with $\psi_{\tau}(x)=\tau-\mbox{I}(x<0)$, it follows that
\begin{flalign}\label{a.4}
Q_n(\theta_{\tau})&=-
\f{1}{n}\sum_{t=1}^n\left[\f{\theta_{\tau}'\tilde{z}_t}{\sigma_t^{\delta}(\tilde{\theta}_{n,r})}-b_{\tau}\right]
\psi_{\tau}\left(\f{y_t}{\sigma_t^{\delta}(\tilde{\theta}_{n,r})}
-b_{\tau}\right)\nonumber\\
&\quad\quad+\f{1}{n}\sum_{t=1}^n\int_0^{\f{\theta_{\tau}'\tilde{z}_t}{\sigma_t^{\delta}(\tilde{\theta}_{n,r})}-b_{\tau}}
\mbox{I}\left(\f{y_t}{\sigma_t^{\delta}(\tilde{\theta}_{n,r})}
\leq s+b_{\tau}\right)
-\mbox{I}\left(\f{y_t}{\sigma_t^{\delta}(\tilde{\theta}_{n,r})}
\leq b_{\tau}\right)ds\nonumber\\
&\equiv -I_{11}(\theta_\tau)+I_{12}(\theta_\tau).
\end{flalign}

Next, we consider $I_{11}(\theta_\tau)$. By Proposition 2.1 in Francq and Zako\"{i}an (2013a), $h_t\to\infty$ as $t\to\infty$, and hence
\begin{flalign}\label{a.5}
\left|\f{h_{t-1}}{h_t}-\f{1}{a_0(\eta_{t-1})}\right|=\left|\f{-\omega_0}{a_0(\eta_{t-1})h_t}\right|\to0\mbox{ as }t\to\infty.
\end{flalign}
By Lemma \ref{lema1}(i), it follows that
\begin{flalign}\label{a.6}
\sup_{\theta\in \Theta_0^*}\left|\frac{\sigma_{t-1}^{\delta}(\theta)}{h_t}-\f{v_{t-1}(\vartheta)}{a_0(\eta_{t-1})}\right|\rightarrow 0 \mbox{ a.s. as $t\rightarrow \infty$}.
\end{flalign}
Define $Z_{t}(\theta)=(1, (\epsilon_{t-1}^+)^{\delta},(-\epsilon_{t-1}^-)^{\delta},\sigma^{\delta}_{t-1}(\theta))'$
and $\varsigma_{t}(\vartheta)=
\big(0,\f{(\eta_{t-1}^+)^{\delta}}{a_0(\eta_{t-1})},\frac{(-\eta_{t-1}^-)^{\delta}}{a_0(\eta_{t-1})},\frac{v_{t-1}(\vartheta)}{a_0(\eta_{t-1})}\big)'$.
Since $(\epsilon_{t-1}^+)^{\delta}/h_{t-1}=(\eta_{t-1}^+)^{\delta}$ and
$(-\epsilon_{t-1}^-)^{\delta}/h_{t-1}=(-\eta_{t-1}^-)^{\delta}$, by (\ref{a.5})-(\ref{a.6}) we have
\begin{flalign}\label{a.7}
\sup_{\theta\in \Theta_0^*}\left\|\frac{Z_{t}(\theta)}{h_t}-\varsigma_{t}(\vartheta)\right\|\rightarrow 0 \mbox{ a.s. as $t\rightarrow \infty$}.
\end{flalign}
Note that $\tilde{z}_{t}=Z_{t}(\tilde{\theta}_{n,r})$ and $y_t=T(\eta_t)h_t$. Then, it is not difficult to have
\begin{flalign}\label{a.8}
I_{11}(\theta_\tau)&=\f{1}{n}\sum_{t=1}^n\left[\f{\theta_{\tau}'\tilde{z}_t/h_t}{\sigma_t^{\delta}(\tilde{\theta}_{n,r})/h_t}-b_\tau\right]
\psi_{\tau}\left(\f{T(\eta_t)}{\sigma_t^{\delta}(\tilde{\theta}_{n,r})/h_t}
-b_\tau\right) \nonumber\\
&=\f{1}{n}\sum_{t=1}^n\left[\f{\theta_{\tau}'\varsigma_t(\tilde{\vartheta}_{n,r})}{v_t(\tilde{\vartheta}_{n,r})}-b_\tau\right]
\psi_{\tau}\left(\f{T(\eta_t)}{\sigma_t^{\delta}(\tilde{\theta}_{n,r})/h_t}
-b_\tau\right)+o_{p}(1) \nonumber\\
&=\f{1}{n}\sum_{t=1}^n\left[\f{\theta_{\tau}'\varsigma_t(\vartheta_{0})}{v_t(\vartheta_{0})}-b_\tau\right]
\psi_{\tau}\left(\f{T(\eta_t)}{\sigma_t^{\delta}(\tilde{\theta}_{n,r})/h_t}
-b_\tau\right)+o_{p}(1),
\end{flalign}
where the second equality holds by Lemma \ref{lema1}(i), (\ref{a.7}) and the boundedness of $\psi_{\tau}(\cdot)$, and the last equality holds
by Taylor's expansion, Lemma \ref{lema3}(ii), and the fact that
$$\sup_{\theta\in\Theta_0} \left|\frac{1}{n}\sum_{t=1}^{n}\frac{\partial}{\partial\vartheta}
\left(\frac{\varsigma_t(\vartheta)}{v_t(\vartheta)}\right)\right|=O_{p}(1).$$
Furthermore, by the double expectation, Lemma \ref{lema1}(i), Assumption \ref{asm3.2}, and standard arguments for tightness,
we can prove
\begin{flalign}\label{a.9}
&\sup_{\theta\in\Theta}\left|\f{1}{n}\sum_{t=1}^n\left[\f{\theta_{\tau}'\varsigma_t(\vartheta_{0})}{v_t(\vartheta_{0})}-b_\tau\right]
\left[\psi_{\tau}\left(\f{T(\eta_t)}{\sigma_t^{\delta}(\theta)/h_t}
-b_\tau\right)-\psi_{\tau}\left(\f{T(\eta_t)}{v_{t}(\vartheta)}
-b_\tau\right)\right]\right| \nonumber\\
&\quad=o_{p}(1).
\end{flalign}
Hence, by (\ref{a.8}) and (\ref{a.9}), it follows that
\begin{flalign}\label{a.10}
I_{11}(\theta_\tau)&=\f{1}{n}\sum_{t=1}^n\left[\f{\theta_{\tau}'\varsigma_t(\vartheta_{0})}{v_t(\vartheta_{0})}-b_\tau\right]
\psi_{\tau}\left(\f{T(\eta_t)}{v_t(\tilde{\vartheta}_{n,r})}
-b_\tau\right)+o_{p}(1) \nonumber\\
&=E\left\{\left[\f{\theta_{\tau}'\varsigma_t(\vartheta_{0})}{v_t(\vartheta_{0})}-b_\tau\right]
\psi_{\tau}\left(\f{T(\eta_t)}{v_t(\tilde{\vartheta}_{n,r})}
-b_\tau\right)\right\}+o_{p}(1) \nonumber\\
&=E\left\{\left[\f{\theta_{\tau}'\varsigma_t(\vartheta_0)}{v_t(\vartheta_{0})}-b_\tau\right]
\psi_{\tau}\left(\f{T(\eta_t)}{v_t(\vartheta_{0})}
-b_\tau\right)\right\}+o_{p}(1) \nonumber \\
&=E\left\{\left[\vartheta_{\tau}'\xi_t-b_{\tau}\right]
\psi_{\tau}\left(T(\eta_t)
-b_{\tau}\right)\right\}+o_{p}(1) \nonumber\\
&=o_{p}(1),
\end{flalign}
where the second equality holds by the uniform ergodic theorem, the third equality
holds by the dominated convergence theorem and Lemma \ref{lema3}(ii), the fourth equality holds since $v_{t}(\vartheta_0)=1$ and $\varsigma_{t}(\vartheta_0)=(0,\xi_t)$, and the last equality holds by the double expectation and
the fact that the $\tau$th quantile of $T(\eta_t)$ is $b_{\tau}$.

Third, we consider $I_{12}(\theta_\tau)$. As for (\ref{a.10}), we can show
\begin{flalign}\label{a.11}
I_{12}(\theta_\tau)&=E\left\{\int_0^{\vartheta_{\tau}'\xi_t-b_{\tau}}
\mbox{I}\left(T(\eta_t)
\leq s+b_{\tau}\right)
-\mbox{I}\left(T(\eta_t)
\leq b_{\tau}\right)ds\right\}+o_{p}(1) \nonumber\\
&=E\left\{\int_0^{\vartheta_{\tau}'\xi_t-\vartheta_{\tau0}'\xi_t}
[f(\breve{\vartheta}_{\tau})]s\,ds\right\}+o_{p}(1) \nonumber\\
&\equiv H(\vartheta_{\tau})+o_{p}(1),
\end{flalign}
where $\breve{\vartheta}_{\tau}$ lies between $s+b_{\tau}$ and $b_{\tau}$, and the second equality holds by the double expectation, Taylor's expansion, and the fact that $b_{\tau}=\vartheta_{\tau0}'\xi_t$.

Note that $|\breve{\vartheta}_{\tau}|\leq |b_{\tau}|+|(\vartheta_{\tau}-\vartheta_{\tau 0})'\xi_t|\leq C_0$ for some constant $C_0>0$.
By (\ref{a.4}), (\ref{a.10}) and (\ref{a.11}), we have that $Q_{n}(\theta_{\tau})=H(\vartheta_{\tau})+o_{p}(1)$, where
$$H(\vartheta_{\tau})
\geq(\vartheta_{\tau}-\vartheta_{\tau 0})'E\left\{\frac{[\inf_{|x|\leq C_0}f(x)]}{2}\xi_{t}\xi_{t}'\right\}(\vartheta_{\tau}-\vartheta_{\tau 0}),
$$
and the equality holds if and only if $\vartheta_{\tau}=\vartheta_{\tau 0}$.
Hence, the proof of (i) is completed by standard arguments, invoking the compactness of $\Theta_{\tau}$.
\end{proof}

Write $\tilde{z}_t=(1, \tilde{z}_{t,\vartheta}')'$, where
$\tilde{z}_{t,\vartheta}=((\epsilon_{t-1}^+)^{\delta},(-\epsilon_{t-1}^-)^{\delta},\sigma^{\delta}_{t-1}(\tilde{\theta}_{n,r}))'$.
Define $\bar{z}_t=(1, \bar{z}_{t,\vartheta}')'$,  where
$\bar{z}_{t,\vartheta}=((\epsilon_{t-1}^+)^{\delta},(-\epsilon_{t-1}^-)^{\delta},\sigma^{\delta}_{t-1}(\theta_0))'.$

\begin{lem}\label{lema5}
Suppose that Assumptions \ref{asm3.1}-\ref{asm3.2} hold and $E|\eta_t|^{2r} <\infty$.

(i) If $\gamma_0>0$, and $P(\eta_t=0)=0$, then
\begin{flalign}
&I_{2}=o_{p}(1),\,\,\,\,\,\,I_{3}=-f(b_{\tau})b_{\tau}\Gamma_{\vartheta}[\sqrt{n}(\tilde{\vartheta}_{n,r}-\vartheta_0)]+o_{p}(1), \label{a.12}\\
&\mbox{ and }I_{4}=[-f(b_{\tau})\Omega_{\vartheta}+o_{p}(1)][\sqrt{n}(\hat{\vartheta}_{\tau n,r}-\vartheta_{\tau 0})]+o_{p}(1),  \label{a.13}
\end{flalign}
where
\begin{flalign*}
I_{2}&=\frac{1}{\sqrt{n}}\sum_{t=1}^n\psi_{\tau}\left(y_t-\hat{\theta}_{\tau n,r}'\tilde{z}_t\right)
\left[\f{\tilde{z}_{t,\vartheta}}{\sigma^{\delta}_t(\tilde{\theta}_{n,r})}-
\f{\bar{z}_{t,\vartheta}}{\sigma^{\delta}_t(\theta_{0})}\right],\\
I_{3}&=\frac{1}{\sqrt{n}}\sum_{t=1}^n\left[\psi_{\tau}\left(y_t-\hat{\theta}_{\tau n,r}'\tilde{z}_t\right)-
\psi_{\tau}\left(y_t-\hat{\theta}_{\tau n,r}'\bar{z}_t\right)\right]
\f{\bar{z}_{t,\vartheta}}{\sigma^{\delta}_t(\theta_{0})},\\
I_{4}&=\frac{1}{\sqrt{n}}\sum_{t=1}^n\left[\psi_{\tau}\left(y_t-\hat{\theta}_{\tau n,r}'\bar{z}_t\right)-
\psi_{\tau}\left(y_t-\theta_{\tau0}'\bar{z}_t\right)\right]
\f{\bar{z}_{t,\vartheta}}{\sigma^{\delta}_t(\theta_{0})}.
\end{flalign*}

(ii)  If $\gamma_0=0$,  $P(\eta_t=0)=0$, $\beta<\|1/a_0(\eta_t)\|_p^{-1}$ for any $\theta\in \Theta$ and  some $p>1$, and
Assumption \ref{asm3.3} is satisfied,
then (\ref{a.12})-(\ref{a.13}) hold.
\end{lem}

\begin{proof}
We only show the proof of (i), and the proof of (ii) is similar.

First, we consider $I_2$. Without loss of generality, we only show that $I_{21}=o_{p}(1)$, where $I_{21}$ is the first entry of $I_2$.
Note that
\begin{flalign}\label{a.14}
I_{21}&=\frac{1}{\sqrt{n}}\sum_{t=1}^n\psi_{\tau}\left(y_t-\hat{\theta}_{\tau n,r}'\tilde{z}_t\right)
(\eta_{t-1}^+)^{\delta}h_{t-1}\left[\f{1}{\sigma^{\delta}_t(\tilde{\theta}_{n,r})}-
\f{1}{\sigma^{\delta}_t(\theta_{0})}\right] \nonumber\\
&=\frac{1}{\sqrt{n}}\sum_{t=1}^n\psi_{\tau}\left(y_t-\hat{\theta}_{\tau n,r}'\tilde{z}_t\right)
\f{(\eta_{t-1}^+)^{\delta}h_{t-1}}{\sigma^{2\delta}_t(\check{\theta}_{n,r})}
\frac{\partial \sigma^{\delta}_t(\check{\theta}_{n,r})}{\partial\vartheta'}(\tilde{\vartheta}_{n,r}-\vartheta_0) \nonumber\\
&\quad+
\frac{1}{\sqrt{n}}\sum_{t=1}^n\psi_{\tau}\left(y_t-\hat{\theta}_{\tau n,r}'\tilde{z}_t\right)
\f{(\eta_{t-1}^+)^{\delta}h_{t-1}}{\sigma^{2\delta}_t(\check{\theta}_{n,r})}
\frac{\partial \sigma^{\delta}_t(\check{\theta}_{n,r})}{\partial\omega}(\tilde{\omega}_{n,r}-\omega_0) \nonumber\\
&\equiv I_{21,1}+I_{21,2}.
\end{flalign}

By a similar argument for Lemma 7.5 in Francq and Zako\"{i}an (2013a), we can show that $I_{21,2}=o_{p}(1)$. For $I_{21,1}$, since
$\sqrt{n}(\tilde{\vartheta}_{n,r}-\vartheta_0)=O_{p}(1)$ by Lemma \ref{lema3}(ii),
we have
\begin{flalign}\label{a.15}
I_{21,1}&=\frac{1}{n}\sum_{t=1}^n\psi_{\tau}\left(y_t-\hat{\theta}_{\tau n,r}'\tilde{z}_t\right)
\f{(\eta_{t-1}^+)^{\delta}[h_{t-1}/h_{t}]}{[\sigma^{\delta}_t(\check{\theta}_{n,r})/h_{t}]^2}
\frac{1}{h_t}\frac{\partial \sigma^{\delta}_t(\check{\theta}_{n,r})}{\partial\vartheta'}[\sqrt{n}(\tilde{\vartheta}_{n,r}-\vartheta_0)] \nonumber\\
&=\frac{1}{n}\sum_{t=1}^n\psi_{\tau}\left(y_t-\hat{\theta}_{\tau n,r}'\tilde{z}_t\right)
\f{(\eta_{t-1}^+)^{\delta}}{a_{0}(\eta_{t-1})}\frac{d_{t}(\vartheta_0)}{
[v_{t}(\vartheta_0)]^2}[\sqrt{n}(\tilde{\vartheta}_{n,r}-\vartheta_0)]+o_{p}(1) \nonumber\\
&\equiv I_{21,1}^{\dag}[\sqrt{n}(\tilde{\vartheta}_{n,r}-\vartheta_0)]+o_{p}(1),
\end{flalign}
where the second equality holds by Lemmas \ref{lema1}(i) and \ref{lema2}(i) and the similar arguments as for (\ref{a.8}) and (\ref{a.10}).

Write
$\psi_{\tau}\left(y_t-\hat{\theta}_{\tau n,r}'\tilde{z}_t\right)=\psi_{\tau}\left(T(\eta_t)-b_{\tau}+c_{\tau,nt}\right)$,
where $c_{\tau,nt}=b_{\tau}-\hat{\theta}_{\tau n,r}'\tilde{z}_t/h_t$.
Since the $\tau$th quantile of $T(\eta_t)$ is $b_{\tau}$, by the ergodic theorem we have
\begin{flalign*}
I_{21,1}^{\dag}&=\frac{1}{n}\sum_{t=1}^n\left[\psi_{\tau}\left(T(\eta_t)-b_{\tau}+c_{\tau,nt}\right)
-\psi_{\tau}\left(T(\eta_t)-b_{\tau}\right)\right]
\f{(\eta_{t-1}^+)^{\delta}}{a_{0}(\eta_{t-1})}\frac{d_{t}(\vartheta_0)}{
[v_{t}(\vartheta_0)]^2}+o_{p}(1) \nonumber\\
&=\frac{1}{n}\sum_{t=1}^n \chi_{t}(c_{\tau,nt})+o_{p}(1),
\end{flalign*}
where
$$\chi_{t}(x)=\left[\psi_{\tau}\left(T(\eta_t)-b_{\tau}+x\right)
-\psi_{\tau}\left(T(\eta_t)-b_{\tau}\right)\right]
\f{(\eta_{t-1}^+)^{\delta}}{a_{0}(\eta_{t-1})}\frac{d_{t}(\vartheta_0)}{
[v_{t}(\vartheta_0)]^2}.$$
By Lemmas \ref{lema1}(i), \ref{lema2}(i), \ref{lema3}(ii) and \ref{lema4}(i), we know that
$c_{\tau,nt}=o_{p}(1)$ for sufficient large $t$. Hence, for any $\varepsilon, \eta>0$, there exits a $t_0(\varepsilon)>0$
such that
\begin{flalign}\label{a.16}
P\left(|c_{\tau,nt}|> \eta\right)<\frac{\varepsilon}{2}
\end{flalign}
for $t\geq t_0$, and
\begin{flalign}\label{a.17}
I_{21,1}^{\dag}=\frac{1}{n}\sum_{t=t_0}^n \chi_{t}(c_{\tau,nt})+o_{p}(1).
\end{flalign}
Note that
$
\sup_{|x|\leq \eta}\left|\frac{1}{n}\sum_{t=t_0}^n\chi_{t}(x)\right|\leq \sup_{|x|\leq \eta}|\chi_{t}(x)|
$
and
$\lim_{\eta\to0}E(\sup_{|x|\leq \eta}|\chi_{t}(x)|)=0$ by the double expectation and dominated convergence theorem. Thus,
by Markov's inequality,  for any $\varepsilon, \varepsilon'>0$, there exists a $\eta_0(\varepsilon)>0$ such that
$P(\sup_{|x|\leq \eta_0}\left|\frac{1}{n}\sum_{t=t_0}^n\chi_{t}(x)\right|>\varepsilon')<\varepsilon/2$. By (\ref{a.16}), it follows that
\begin{flalign*}
P\left(\left|\frac{1}{n}\sum_{t=t_0}^n \chi_{t}(c_{\tau,nt})\right|>\varepsilon'\right)
&\leq
P\left(\left|\frac{1}{n}\sum_{t=t_0}^n \chi_{t}(c_{\tau,nt})\right|>\varepsilon', |c_{\tau,nt}|\leq \eta_0\right)
+P\left(|c_{\tau,nt}|> \eta_0\right) \nonumber\\
&\leq
P\left(\sup_{|x|\leq \eta_0}\left|\frac{1}{n}\sum_{t=t_0}^n\chi_{t}(x)\right|>\varepsilon'\right)+\frac{\varepsilon}{2} \nonumber\\
&\leq \frac{\varepsilon}{2}+\frac{\varepsilon}{2}=\varepsilon,
\end{flalign*}
which implies that $I_{21,1}^{\dag}=o_{p}(1)$ by (\ref{a.17}), $I_{21,1}=o_{p}(1)$ by (\ref{a.15}), and $I_{21}=o_{p}(1)$ by (\ref{a.14}).

Second, by Lemmas \ref{lema1}(i), \ref{lema2}(i), \ref{lema3}(ii) and \ref{lema4}(i), Proposition 2.1 in Francq and Zako\"{i}an (2013a), and a similar argument as for Theorem 2.1 in Zheng et al. (2018), we can prove the result for $I_{3}$.

Third, we consider $I_4$. Let
\begin{flalign*}
\upsilon_t(\omega,u)&=\left[\psi_{\tau}\left(\frac{y_t-u'\bar{z}_{t,\vartheta}}{h_t}-\f{\omega+\vartheta_{\tau 0}'\bar{z}_{t,\vartheta}}{h_t}\right)-\psi_{\tau}\left(\frac{y_t-{\vartheta}_{\tau 0}'\bar{z}_{t,\vartheta}}{h_t}-\f{\omega_{\tau0}}{h_t}\right)\right]\f{\bar{z}_{t,\vartheta}}{\sigma^{\delta}_t(\theta_{0})}\\
&=\left[\mbox{I}\left(T(\eta_t)<\f{\vartheta_{\tau 0}'\bar{z}_{t,\vartheta}+\omega_{\tau 0}}{h_t}\right)-\mbox{I}\left(T(\eta_t)<\frac{u'\bar{z}_{t,\vartheta}}{h_t}+\f{\omega+\vartheta_{\tau 0}'\bar{z}_{t,\vartheta}}{h_t}\right)\right]\f{\bar{z}_{t,\vartheta}}{\sigma^{\delta}_t(\theta_{0})}.
\end{flalign*}
Then, we can see that
$
I_{4}=\f{1}{\sqrt{n}}\sum_{t=1}^n \upsilon_t(\hat{\omega}_{\tau n,r},\hat{u}_{\tau n,r}),
$
where $\hat{u}_{\tau n,r}=\hat{\vartheta}_{\tau n,r}-\vartheta_{\tau 0}$. Since $\mbox{I}(\cdot)$ is an increasing function and
$\underline{\omega}_{\tau}\leq \hat{\omega}_{\tau n}\leq \overline{\omega}_{\tau}$ for some constants
$\underline{\omega}_{\tau}$ and $\overline{\omega}_{\tau}$, we only need to show
\begin{flalign}\label{a.18}
\f{1}{\sqrt{n}}\sum_{t=1}^n \upsilon_t(\omega,\hat{u}_{\tau n,r})=[-f(b_{\tau})\Omega_{\vartheta}+o_{p}(1)](\sqrt{n}\hat{u}_{\tau n,r})+o_{p}(1)
\end{flalign}
for any fixed $\omega$. Rewrite
\begin{flalign}\label{a.19}
\f{1}{\sqrt{n}}\sum_{t=1}^n \upsilon_t(\omega,u)=W_{n}(\omega, u)+S_{n}(\omega, u),
\end{flalign}
where
\begin{flalign*}
W_{n}(\omega,u)&=\f{1}{\sqrt{n}}\sum_{t=1}^nE[\upsilon_t(\omega,u)|\mathcal{F}_{t-1}],\\
S_n(\omega,u)&=\f{1}{\sqrt{n}}\sum_{t=1}^n\left\{\upsilon_t(\omega,u)-E[\upsilon_t(\omega,u)|\mathcal{F}_{t-1}]\right\}.
\end{flalign*}
By Assumptions \ref{asm3.1}-\ref{asm3.2}, Lemmas \ref{lema1}(i) and \ref{lema4}(i), and Proposition 2.1 in Francq and Zako\"{i}an (2013a)  it is not difficult to show that  $W_{n}(\omega,u)=-f(b_{\tau})\Omega_{\vartheta}(\sqrt{n}u)+o_{p}(1)$. Meanwhile,
by a similar argument as for Lemma 2.2 in Zhu and Ling (2011), we can show that for fixed $\omega$ and any $\eta>0$, we have
\begin{flalign*}
\sup_{\|u\|\leq \eta}\f{\|S_n(\omega,u)\|}{1+\sqrt{n}\|u\|}=o_p(1),
\end{flalign*}
which implies that $S_{n}(\omega,\hat{u}_{\tau n,r})=o_{p}(\sqrt{n}\hat{u}_{\tau n,r})+o_{p}(1)$ by Lemma \ref{lema4}(i).
Hence, by (\ref{a.19}) it follows that
$$\f{1}{\sqrt{n}}\sum_{t=1}^n \upsilon_t(\omega,\hat{u}_{\tau n,r})=
-f(b_{\tau})\Omega_{\vartheta}(\sqrt{n}\hat{u}_{\tau n,r})+o_{p}(\sqrt{n}\hat{u}_{\tau n,r})+o_{p}(1),$$
i.e., (\ref{a.18}) holds. This completes all of the proofs.
\end{proof}

\begin{lem}\label{lema6}
Suppose that Assumptions \ref{asm3.1}-\ref{asm3.2} hold and $E|\eta_t|^{2r} <\infty$.

(i) If $\gamma_0>0$, and $P(\eta_t=0)=0$, then
\begin{flalign}\label{a.20}
I_{5}&=\frac{1}{\sqrt{n}}\sum_{t=1}^n\psi_{\tau}\left(\eta_t-Q_{\tau,\eta}\right)
\f{\bar{z}_{t,\vartheta}}{\sigma^{\delta}_t(\theta_{0})}+o_{p}(1)\nonumber\\
&\to_{d} N(0, (\tau-\tau^2)E(\xi_{t}\xi_{t}'))\,\,\mbox{ as }n\to\infty,
\end{flalign}
where
\begin{flalign*}
I_{5}&=
\frac{1}{\sqrt{n}}\sum_{t=1}^n\psi_{\tau}\left(y_t-\theta_{\tau 0}'\bar{z}_t\right)
\f{\bar{z}_{t,\vartheta}}{\sigma^{\delta}_t(\theta_{0})}.
\end{flalign*}

(ii)  If $\gamma_0=0$,  $P(\eta_t=0)=0$, $\beta<\|1/a_0(\eta_t)\|_p^{-1}$ for any $\theta\in \Theta$ and  some $p>1$, and
Assumption \ref{asm3.3} is satisfied, then (\ref{a.20}) holds.
\end{lem}

\begin{proof}
The proof can be accomplished by following a similar argument as for Lemma 7.4 in Francq and Zako\"{i}an (2013a).
\end{proof}

\textsc{Proof Theorem \ref{thm3.1}}. (i) Following the proofs in Zheng et al. (2018) and Hamadeh and Zako\"{i}an (2011), we can show
\begin{flalign}\label{a.21}
\sqrt{n}(\hat{\theta}_{\tau n,r}-\theta_{\tau 0})&=\frac{\Omega^{-1}}{f(b_{\tau})}
\left[\frac{1}{\sqrt{n}}\sum_{t=1}^n\psi_{\tau}\left(\eta_t-Q_{\tau,\eta}\right)
\f{z_{t}}{h_t(\theta_{0})}\right]\nonumber\\
&\quad-b_{\tau}\Omega^{-1}\Gamma[\sqrt{n}(\tilde{\theta}_{n,r}-\theta_0)]+o_{p}(1),
\end{flalign}
which entails (i) by Lemma \ref{lema3}(i) and standard arguments.

(ii) Following the same argument as for Theorem 2.1 in Francq and Zako\"{i}an (2012), the subgradient
derivative with respect to $\vartheta_{\tau}$ is asymptotically equal to zero at the minimum,
since $\hat{\vartheta}_{\tau n,r}\to_{p} \vartheta_{\tau 0}$ by Lemma \ref{lema4}(i), and $\vartheta_{\tau 0}$ belongs to the interior
of $\Theta_{\tau}$. This implies
\begin{flalign}\label{a.22}
0=\frac{1}{\sqrt{n}}\sum_{t=1}^n\psi_{\tau}\left(y_t-\hat{\theta}_{\tau n,r}'\tilde{z}_t\right)
\f{\tilde{z}_{t,\vartheta}}{\sigma^{\delta}_t(\tilde{\theta}_{n,r})}.
\end{flalign}
Moreover, by Lemmas \ref{lema5}(i) and \ref{lema6}(i), we have
\begin{flalign*}
&\frac{1}{\sqrt{n}}\sum_{t=1}^n\psi_{\tau}\left(y_t-\hat{\theta}_{\tau n,r}'\tilde{z}_t\right)
\f{\tilde{z}_{t,\vartheta}}{\sigma^{\delta}_t(\tilde{\theta}_{n,r})}\nonumber\\
&=I_2+I_3+I_{4}+I_{5} \nonumber\\
&=-f(b_{\tau})b_{\tau}\Gamma_{\vartheta}[\sqrt{n}(\tilde{\vartheta}_{n,r}-\vartheta_0)]+
[-f(b_{\tau})\Omega_{\vartheta}+o_{p}(1)][\sqrt{n}(\hat{\vartheta}_{\tau n,r}-\vartheta_{\tau 0})]\nonumber\\
&\quad+\frac{1}{\sqrt{n}}\sum_{t=1}^n\psi_{\tau}\left(\eta_t-Q_{\tau,\eta}\right)
\f{\bar{z}_{t,\vartheta}}{\sigma^{\delta}_t(\theta_{0})}+o_{p}(1).
\end{flalign*}
By (\ref{a.22}), it follows that
\begin{flalign}\label{a.23}
\sqrt{n}(\hat{\vartheta}_{\tau n,r}-\vartheta_{\tau 0})&=\frac{\Omega_{\vartheta}^{-1}}{f(b_{\tau})}
\left[\frac{1}{\sqrt{n}}\sum_{t=1}^n\psi_{\tau}\left(\eta_t-Q_{\tau,\eta}\right)
\f{\bar{z}_{t,\vartheta}}{\sigma^{\delta}_t(\theta_{0})}\right]\nonumber\\
&\quad-b_{\tau}\Omega_{\vartheta}^{-1}\Gamma_{\vartheta}[\sqrt{n}(\tilde{\vartheta}_{n,r}-\vartheta_0)]+o_{p}(1),
\end{flalign}
which implies (ii) holds by Lemmas \ref{lema3}(ii) and \ref{lema6}(i), and standard arguments.

(iii) Its proof can be accomplished by following a similar argument as for (ii). This completes all of the proofs. $\hfill\square$


\end{document}